\documentclass[aps,prx,reprint,superscriptaddress,nofootinbib,floatfix]{revtex4-2}

\usepackage[T1]{fontenc}
\usepackage[utf8]{inputenc}
\usepackage{amsmath,amssymb,mathtools,bm}
\usepackage{amsthm}
\usepackage{graphicx}
\usepackage{booktabs}
\usepackage{array}
\usepackage{microtype}
\usepackage{hyperref}
\usepackage{xcolor}
\DeclareMathOperator{\bc}{bc}
\usepackage{natbib}
\usepackage{comment}
\usepackage{tikz,pgfplots}
\usepackage{subcaption}
\usepackage{adjustbox}

\hypersetup{
colorlinks=true,
linkcolor=black,
citecolor=black,
urlcolor=black,
pdftitle={Exact resource laws for passive entanglement networks}
}

\newtheorem{theorem}{Theorem}
\newtheorem{lemma}{Lemma}
\newtheorem{proposition}{Proposition}
\newtheorem{corollary}{Corollary}
\newtheorem{definition}{Definition}

\newcommand{\erf}{\operatorname{erf}}

\newcommand{\Oset}{\mathcal O}
\newcommand{\Rset}{\mathcal R}
\newcommand{\Kset}{\mathcal K}
\newcommand{\Dset}{\mathcal D}
\newcommand{\Sset}{\mathcal S}
\newcommand{\Bic}{\operatorname{Bic}}
\newcommand{\eps}{\varepsilon}
\newcommand{\Rconj}{R_{\rm conj}}
\newcommand{\etast}{\eta_{\rm st}}

\newif\ifshowchanges
\ifdefined\CLEANVERSION\showchangesfalse\else\showchangestrue\fi
\ifshowchanges
  \newenvironment{revision}{\color{blue}}{}
  
\else

\fi
\begin{document}

\title{Exact Resource Laws for Passive Wavelength Routing in Entanglement Networks}

\author{Ekta Panwar}
\thanks{ekta.panwar@savba.sk}
\author{Gilberto Borges}
\author{Saeid Salari}
\author{Kartik Kakade}
\author{Samgeeth Puliyil}
\author{Peter Rap\v{c}an}
\author{Mario Ziman}
\author{Djeylan Aktas}

\affiliation{Institute of Physics, Slovak Academy of Sciences, 841 04 Bratislava, Slovakia}

\date{September 24, 2026}

\begin{abstract}

Entanglement-based networks provide a scalable framework for multiuser quantum communication by passively routing spectrally correlated photon pairs across interconnected nodes. Several wavelength-allocation schemes have already been demonstrated experimentally, but these designs do not yet provide a general way to determine how spectral use, receiver load, repeated connections, fan-out, and optical loss constrain one another. Motivated by a field-deployed four-node network, we formulate wavelength assignment as an exact resource-optimization problem on the network connectivity graph. We consider uniquely paired wavelength channels and require that no user receive both channels of a pair. For one-sided fan-out, we obtain an exact optimization for arbitrary networks, and closed-form resource laws for complete networks and for networks with one excluded partner per user. Allowing both conjugate wavelengths to fan out changes the resource landscape. In complete networks, a balanced binary hierarchy uses the minimum number of spectral layers among nonredundant assignments while reducing the maximum receiver wavelength load to a quantity that grows logarithmically with the number of users. For eight users, the source-cost ordering depends on splitter transmission: changing the transmission can change which architecture requires less pair generation. We formulate a BBM92 rate model for evaluating the corresponding architectures. For a single continuous-wave broadband source, the occupied wavelength pairs are generated with fixed relative weights, while detector singles and accidental coincidences couple the link rates. The same resource laws can therefore be evaluated using the measured parameters of the hardware that implements the network.

\end{abstract}

\maketitle

\section{Introduction}
Rapid advancements in second-generation quantum technologies have accelerated the development of quantum networks (QNs), providing a critical framework to interconnect quantum devices. These networks have already enabled real-world quantum key distribution (QKD) deployments~\cite{Neumann2022Transnational,Joshi2020EightUser}, while also supporting distributed quantum sensing~\cite{Malia2022Sensing} and distributed quantum computing~\cite{Main2025DistributedComputing}.

A central design problem is how to realize fully connected, or fully meshed, networks without relying on intermediary trusted nodes. One route to a fully meshed architecture is to distribute entangled photons using the spectral correlations of second-order nonlinear light--matter interactions~\cite{Joshi2020EightUser,Wengerowsky2018Network}. These configurations leverage standard telecommunication technologies for spectral multiplexing (MUX) and demultiplexing (DEMUX), hereafter referred to as wavelength-demultiplexed quantum networks (DEMUX-QNs). The broadband spectrum is demultiplexed into discrete wavelength channels, and conjugate photon pairs are routed among network users. For a single-pump down-conversion source, the paired frequencies are symmetric about half the pump frequency; the four-wave-mixing source in Ref.~\cite{Liu2022FortyUser} has a different energy-matching condition~\cite{Neumann2022Transnational,Joshi2020EightUser,Wengerowsky2018Network}.

Following the pioneering demonstration of a DEMUX-QN~\cite{Wengerowsky2018Network}, subsequent architectures have improved scalability without relying on trusted nodes~\cite{Joshi2020EightUser,Liu2022FortyUser}, while theoretical models have been developed to characterize such networks~\cite{Bathaee2023EntangledWDM,Bathaee2023UnifiedWDM}. Spectrum-allocation methods also address flexible-grid entanglement networks~\cite{Alnas2022ResourceAllocation,Goisman2026FlexGrid}. Here, we ask a different question: which combinations of spectral-layer count, receiver-wavelength load, repetition, and fan-out are attainable under passive sharing of uniquely paired wavelength channels, with each user receiving at most one channel from each pair? The present manuscript addresses this resource-allocation challenge by
building on the wavelength-assignment formulation introduced for the
complete-graph deployment in Ref.~\cite{companionB} and generalizing it
to arbitrary requested graphs. We represent each correlated wavelength
pair as a graph layer, which we call a spectral layer.


The simplest way to use these layers is to assign one spectral layer to every requested user link~\cite{Wengerowsky2018Network}. The two conjugate wavelengths are routed directly to the two users forming the link, and no wavelength pair is shared with any other connection. Since a complete network of $N$ users has one link between every pair, the number of required spectral layers grows as $N(N-1)/2$.

Instead of reserving a conjugate wavelength pair for a single link, one or both wavelengths can be distributed among several users. If one wavelength is sent to a single user while its conjugate is split among several others, that layer connects one center user to several leaves and forms a star. If both conjugate wavelengths are split, every user receiving one wavelength is connected to every user receiving its conjugate, and the layer forms a biclique. Figure~\ref{fig:optical-overview} illustrates the pairwise assignment, one-sided fan-out, and the resulting wavelength-channel scaling. These architectures have already enabled fully connected entanglement networks without trusted relays or active wavelength switching~\cite{Joshi2020EightUser,Wengerowsky2018Network,Liu2022FortyUser}. Reconfigurable wavelength routing~\cite{Appas2021Flexible,Wang2026Dynamic} and state-multiplexing sources~\cite{Fan2025StateMultiplexing} provide related, but distinct, ways to distribute entanglement.

These demonstrations show that passive sharing is experimentally viable, but the same sharing also divides photon flux among more output branches, can introduce additional splitter loss, increase the wavelength load at a user, or repeat a requested link. For a BBM92 network, these choices affect the pair flux and secret-key rate available to each user pair.

These observations lead to two questions. First, for a prescribed multiuser network, what limits on the number of spectral layers, receiver wavelength load, repeated connections, and fan-out follow from the network topology alone? Second, once splitter loss and prescribed key-rate targets are included, do the architectures that attain these topological limits remain preferable?

To address the first question, we use the graph already implicit in the wavelength assignment. A user becomes a vertex, while a requested pairwise connection becomes an edge. The layer decomposition then directly specifies the spectral demand, receiver wavelength load, fan-out, and repeated connections.

For one-sided fan-out, each requested link is assigned to the endpoint that receives the unsplit wavelength and therefore serves as the center. Links assigned to the same center can be grouped into one spectral layer, up to the allowed fan-out. In graph terms, this choice is an orientation of the requested network, which leads to an exact optimization of the required layer count. We solve the resulting layer--fan-out tradeoff for complete networks and for networks in which every user has one excluded partner. The same construction gives balanced layer sizes, which later determine the extra splitting loss when a layer serves more users.

Allowing both conjugate wavelengths to split changes the problem: one layer can connect two groups of users rather than one center to a set of leaves, separating the layer count from the wavelength load at a receiver. For a complete network, we construct a balanced binary hierarchy that uses the minimum possible number of layers among nonredundant assignments while the maximum receiver load grows only logarithmically with the number of users. For eight users, the hierarchy uses seven layers with maximum receiver load three, whereas every seven-star partition has maximum load seven. When each wavelength is limited to two users, eight layers suffice only if four links are repeated; eliminating these repetitions requires a ninth layer. Thus, minimizing the number of layers does not necessarily minimize receiver load, fan-out, or repeated connections.

We then add the physical cost of key generation and ask whether the topologically optimal architectures remain favorable. We combine graph decomposition with physically realizable splitter configurations and a BBM92 model that relates the delivered entangled-pair flux to the secret-key rate. For prescribed key-rate targets and independently adjustable layer rates, fewer spectral layers or a smaller receiver load need not imply a lower total pair-generation requirement. For eight users, we derive the splitter-transmission threshold at which the seven-layer hierarchy and the eight-layer repeated cover exchange order. We then consider a single CW-pumped broadband source, in which the source spectrum fixes the relative rates of the spectral layers and one detector record may enter several pairwise coincidence analyses within a fan-out layer, coupling the link responses through accidental coincidences.

Star decompositions, biclique covers, and biclique partitions are well-established objects in graph theory~\cite{Tarsi1981Stars,LinShyu1996Stars,CameronHorsley2020Stars,Harangi2025Orientations,GrahamPollak1971,KatonaSzemeredi1967,Shader1993Biclique,Gu2015BoundedBiclique,Pinto2014Biclique}. Here, we use this connection to determine how a passive wavelength assignment fixes concrete network resources and to compare the spectral advantage of an architecture with the optical cost required to realize it.

The remainder of the paper is organized as follows. Section~\ref{sec:routing_algebra}
introduces the optical-layer model and its graph representation. Sections~\ref{sec:one-sided}
and~\ref{sec:two-sided} derive the exact resource laws for one- and two-sided fan-out,
and Section~\ref{sec:k8-frontier} makes the resulting trade-offs explicit for the
eight-user complete network. Sections~\ref{sec:qkdmodel}--\ref{sec:twosided-hardware}
then incorporate splitter loss and key-rate requirements into the resource analysis.
Sections~\ref{sec:cw} and~\ref{sec:published-benchmark} extend the treatment to a
continuous-wave coincidence model and to a common broadband source. Finally,
Section~\ref{sec:discussion} summarizes the main results and discusses their
implications for the design of passive entanglement networks.

\begin{figure*}[t]
\centering
    \hspace*{-2cm}
    \begin{subfigure}{0.24\textwidth}
        \begin{tikzpicture}[scale=0.80,transform shape]
            \node[draw, circle] (A) at (-2,0) {\textbf{A}};
            \node[draw, circle] (B) at (0,2) {\textbf{B}};
            \node[draw, circle] (C) at (2,0) {\textbf{C}};
            \node[draw, circle] (D) at (0,-2) {\textbf{D}};
            \draw[thick,blue] (A) --node[above]{$l_1$} (B);
            \draw[thick,orange] (A) --node[below right]{$l_2$} (C);
            \draw[thick,green!60!black] (A) --node[left,xshift=0.1cm,yshift=-0.1cm]{$l_3$} (D);
            \draw[thick,magenta] (B) --node[above]{$l_4$} (C);
            \draw[thick,violet] (B) --node[above right]{$l_5$} (D);
            \draw[thick,cyan] (C) --node[right,yshift=-0.1cm]{$l_6$} (D);
            \node at (-2,-0.6) {\small Bratislava};
            \node at (2,-0.6) {\small Nitra};
            \node at (0,2.5) {\small Tren\v{c}in};
            \node at (-1,-2) {\small Trnava};
            \node[white] at (0,-2.5) {\small Trnava};
        \end{tikzpicture}
        \subcaption{Pairwise assignment,\newline six links = six spectral layers}
		\end{subfigure}
		\begin{subfigure}{0.24\textwidth}
			\begin{tikzpicture}[scale=0.80,transform shape]
				\node[draw, circle] (A) at (-2,0) {\textbf{A}};
				\node[draw, circle] (B) at (0,2) {\textbf{B}};
				\node[draw, circle] (C) at (2,0) {\textbf{C}};
				\node[draw, circle] (D) at (0,-2) {\textbf{D}};
				\draw[thick,blue] (A) --node[above]{$l_1$} (B);
				\draw[thick,blue] (A) --node[below right]{$l_1$} (C);
				\draw[thick,blue] (A) --node[left,xshift=0.1cm,yshift=-0.1cm]{$l_1$} (D);
                \node[] at (0,-2.6) {};
			\end{tikzpicture}
			\subcaption{\small One-sided fanout\newline one center, three leaves}
		\end{subfigure}
		\begin{subfigure}{0.24\textwidth}
			\begin{tikzpicture}[scale=0.80,transform shape]
				\node[draw, circle] (A) at (-2,0) {\textbf{A}};
				\node[draw, circle] (B) at (0,2) {\textbf{B}};
				\node[draw, circle] (C) at (2,0) {\textbf{C}};
				\node[draw, circle] (D) at (0,-2) {\textbf{D}};
				\draw[thick,blue] (A) --node[above]{$l_1$} (B);
				\draw[thick,orange] (A) --node[below right]{$l_2$} (C);
				\draw[thick,blue] (A) --node[left,xshift=0.1cm,yshift=-0.1cm]{$l_1$} (D);
				\draw[thick,orange] (B) --node[above]{$l_2$} (C);
				\draw[thick,green!60!black] (B) --node[above right]{$l_3$} (D);
				\draw[thick,green!60!black] (C) --node[right,yshift=-0.1cm]{$l_3$} (D);
			\end{tikzpicture}
			\subcaption{Three-star partition of $K_4$\newline six links grouped into three spectral layers}
		\end{subfigure}
		\begin{subfigure}{0.24\textwidth}
			\begin{tikzpicture}[scale=0.65,transform shape]
    		\begin{axis}
    			[axis lines = middle,thick,
    			xlabel = {Number of users $N$},
    			ylabel = {Individual wavelength channels $M_\lambda = 2L$},
    			xlabel style={at={(axis description cs:0.5,-0.125)}, anchor=north},
    			ylabel style={rotate=90,at={(axis description cs:-0.125,0.49)}, anchor=south},
    			xmin = 0, xmax = 33,
    			ymin = 0, ymax = 850,
    			grid = both,
    			legend pos=north west]
    			\addplot[domain=0:30,samples=100,very thick,orange] {x^2-x};
    			\addlegendentry{Pairwise assignment, $2L= N(N-1)$}
                \addplot[domain=0:31,samples=100,very thick,dashed,blue] {2*x-2};
    			\addlegendentry{One-sided, sufficient fanout, $2L= 2(N-1)$}
    		\end{axis}
	        \end{tikzpicture}
			\subcaption{Wavelength-channel scaling}
		\end{subfigure}
\caption{A passive wavelength assignment and its graph representation. (a) Pairwise assignment of the six links of $K_4$, with one spectral layer per requested edge. (b) A one-sided fan-out layer centered at $A$, where one unsplit wavelength is routed to the center and its conjugate is split among three leaves. Panel (c) illustrates a minimum-layer partition; Eq.~\eqref{eq:K4stars} gives a balanced three-star alternative. (d) Scaling of the number of individual wavelength channels, $M_\lambda=2L$, for the pairwise assignment and for the minimum one-sided assignment when sufficient fan-out is available. }
\label{fig:optical-overview}
\end{figure*}

\section{Optical layers and graph representation}
\label{sec:routing_algebra}

We begin with the experimental description of a passive multiuser entanglement network. The source produces pairs of conjugate wavelength channels, and we call each such pair a \emph{spectral layer}. For every spectral layer, the optical network routes one conjugate wavelength channel to a nonempty set of users and the other to a second nonempty set. We then translate this wavelength assignment into a graph-theoretic description and use it throughout the paper to determine the required spectral and receiver resources.

Let $G=(V,E)$ denote the requested entanglement network. The vertex set $V$ represents the $n=|V|$ users, and the edge set $E$ contains the $m=|E|$ requested pairwise links. We take $E\ne\varnothing$; an empty requested network needs no occupied layers. Suppose that the network uses $L$ spectral layers, indexed by
\begin{equation}
\ell\in\{1,\ldots,L\}.
\end{equation}
Each spectral layer contains two conjugate wavelength channels, and each occupied channel belongs to exactly one such pair. Thus the total number of individual wavelength channels is
\begin{equation}
M_\lambda=2L.
\end{equation}
We label these channels by $\mu,\nu\in\{1,\ldots,M_\lambda\}$. For each channel $\mu$, its conjugate partner is denoted by $\bar\mu$, with
\begin{equation}
\overline{\bar\mu}=\mu.
\end{equation}
The unique-partner assumption excludes sources that correlate one physical wavelength channel with several other channels, as in Ref.~\cite{Fan2025StateMultiplexing}. Their spectral cost requires a more general source-correlation model.
We now consider a spectral layer $\ell$. Let its two conjugate wavelength channels be delivered to disjoint, nonempty user sets $A_\ell$ and $B_\ell$:
\begin{equation}
A_\ell\cap B_\ell=\varnothing.
\label{eq:disjoint-sides}
\end{equation}
The disjointness condition means that no user receives both conjugate channels from the same layer. This is a restriction on the routing class, not a general property of passive networks. For example, the intra-subnet construction of Ref.~\cite{Liu2022FortyUser} distributes both conjugate channels to every user in a subnet and is outside this class. The layer-count and receiver-load bounds below assume Eq.~\eqref{eq:disjoint-sides}. Every user in $A_\ell$ is paired with every user in $B_\ell$, while two users on the same side do not receive a conjugate pair from that layer. The realized edge set is the biclique
\begin{equation}
\Bic(A_\ell,B_\ell)
=
\bigl\{\{a,b\}:a\in A_\ell,\ b\in B_\ell\bigr\}.
\label{eq:biclique}
\end{equation}
A one-sided layer is obtained when one of the two sets contains a single user. Writing $A_\ell=\{c_\ell\}$, the layer realizes the star
\begin{equation}
S(c_\ell,B_\ell)
=
\bigl\{\{c_\ell,b\}:b\in B_\ell\bigr\}.
\label{eq:star}
\end{equation}
Thus, every spectral layer considered here is a biclique, and a star is its one-sided special case.

 the family of layer graphs is 
\begin{equation}
\mathcal P
=
\bigl\{\Bic(A_\ell,B_\ell)\bigr\}_{\ell=1}^{L}.
\end{equation}
This family is realized by the $L$ spectral layers. For a requested edge $e\in E$, let $m_{\mathcal P}(e)$ denote the number of spectral layers in which that edge appears. The assignment covers $G$ when every requested edge appears at least once,
\begin{equation}
m_{\mathcal P}(e)\ge 1
\qquad(e\in E),
\end{equation}
and no layer produces an edge outside $E$. The layer graphs then realize exactly the requested network, although some requested links may occur more than once.

We call the cover \emph{nonredundant} when every requested edge appears in exactly one spectral layer:
\begin{equation}
m_{\mathcal P}(e)=1
\qquad(e\in E).
\label{eq:nonredundant}
\end{equation}
In this case, the layer graphs form an edge partition of $G$. For a general cover, the total number of additional edge occurrences is
\begin{equation}
\Delta(\mathcal P)
=
\sum_{e\in E}\bigl[m_{\mathcal P}(e)-1\bigr].
\label{eq:route-overhead}
\end{equation}
The first occurrence realizes the requested link, while every further occurrence contributes one unit to $\Delta(\mathcal P)$. Hence $\Delta(\mathcal P)=0$ exactly when the cover is nonredundant.

The wavelength load at user $u$ is
\begin{equation}
\chi_u(\mathcal P)
=
\bigl|\{\ell:u\in A_\ell\cup B_\ell\}\bigr|,
\qquad
\chi(\mathcal P)=\max_{u\in V}\chi_u(\mathcal P).
\label{eq:receiver-load}
\end{equation}
Thus, $\chi_u(\mathcal P)$ counts the spectral layers, and therefore the wavelength channels, received by user $u$, while $\chi(\mathcal P)$ is the load of the most heavily used receiver. This is a wavelength-incidence count, not a detector count unless the receiver implementation identifies the two. We also quantify the largest splitter side through the maximum side size:
\begin{equation}
\sigma(\mathcal P)
=
\max_\ell\max\{|A_\ell|,|B_\ell|\}.
\label{eq:max-side}
\end{equation}
This is the largest number of users reached by either wavelength channel in any layer. Since $\sigma(\mathcal P)$ records only the larger side, we also specify the ordered layer type $|A_\ell|\times|B_\ell|$. In particular, a $1\times r$ layer is a star, whereas a $p\times q$ layer with $p,q>1$ uses two-sided fan-out.

The same resources can also be read directly from the \textit{connectivity-matrix}
formalism introduced in Ref.~\cite{companionB}, where the product of the
connectivity matrix with the reversal matrix recovers the multiplicity of each
user pair. In the present work, we refer to these objects as the
\textit{channel-delivery matrix} and the \textit{conjugation matrix},
respectively.



\subsection{Equivalence between wavelength and graph descriptions}
\label{subsec:optical-graph-equivalence}
 Define the channel-delivery matrix $Q$ and the conjugation matrix $\Rconj$ by

\begin{align}
Q_{\mu u}=1
&\Longleftrightarrow
\text{channel $\mu$ reaches user $u$},\nonumber\\
(\Rconj)_{\mu\nu}
&=\delta_{\nu,\bar\mu}.
\label{eq:QRdef}
\end{align}
Each row of $Q$ corresponds to one wavelength channel and each column to one user. The matrix $\Rconj$ records which two rows belong to the same spectral layer. Their product
\begin{equation}
C(Q)=Q^{\mathsf T}\Rconj Q
\label{eq:routing-certificate}
\end{equation}
counts conjugate-channel routes between users. The following proposition establishes this equivalence.

\begin{proposition}[Optical--graph equivalence]
\label{prop:optical-graph}
Assume that no user receives both conjugate wavelength channels from the same spectral layer. Then each layer of $Q$ induces the biclique in Eq.~\eqref{eq:biclique}. Moreover,
\begin{align}
L&=\frac{M_\lambda}{2},
\label{eq:resource-L}\\
\chi_u(\mathcal P)&=\sum_{\mu=1}^{M_\lambda}Q_{\mu u},
\label{eq:resource-chi}\\
\Delta(\mathcal P)&=\frac12\sum_{u,v\in V}
\bigl[C_{uv}(Q)-A_{uv}(G)\bigr]
\label{eq:resource-delta}
\end{align}
for every cover of $G$. The assignment is nonredundant if and only if
\begin{equation}
Q^{\mathsf T}\Rconj Q=A(G).
\label{eq:adj-cert}
\end{equation}
\end{proposition}

\begin{proof}
We prove the claim by reading one entry of the routing certificate. For fixed users $u$ and $v$,
\begin{equation}
C_{uv}(Q)
=
\sum_{\mu=1}^{M_\lambda}Q_{\mu u}Q_{\bar\mu v}
\label{eq:routing-entry}
\end{equation}
counts the spectral layers in which $u$ and $v$ receive conjugate wavelength channels. The two recipient sets in one layer therefore generate a biclique. Since these sets are disjoint, a user receives at most one channel from each layer, and the column sum in Eq.~\eqref{eq:resource-chi} is its layer incidence. For a cover, subtracting the adjacency matrix removes the first required occurrence of every requested edge. The remaining symmetric entries count repetitions, and the factor $1/2$ removes double counting. Equation~\eqref{eq:adj-cert} is precisely the condition that every requested edge occurs once and no other edge occurs.
\end{proof}

\subsection{A four-user example}
\label{subsec:four-user-example}

The equivalence is easiest to see in a four-user network. We first group the six requested links into three stars and then read the same three spectral layers from the rows of the channel-delivery matrix.

For users $A,B,C,D$, the complete graph $K_4$ is partitioned by
\begin{equation}
S(A,\{B,D\}),\qquad
S(B,\{C,D\}),\qquad
S(C,\{A,D\}).
\label{eq:K4stars}
\end{equation}
Their edge sets are
\begin{equation}
\{AB,AD\},\qquad
\{BC,BD\},\qquad
\{AC,CD\},
\label{eq:K4groups}
\end{equation}
which contain all six edges exactly once.

Order the six channel rows so that the conjugate pairs are $(1,6)$, $(2,5)$, and $(3,4)$, and take
\begin{equation}
Q_4=
\begin{pmatrix}
1&0&0&0\\
0&1&0&0\\
0&0&1&0\\
1&0&0&1\\
0&0&1&1\\
0&1&0&1
\end{pmatrix},
\label{eq:Q4}
\end{equation}
with columns ordered as $A,B,C,D$. Rows $1$ and $6$ route conjugate channels to $A$ and $\{B,D\}$, respectively; the pairs $(2,5)$ and $(3,4)$ similarly realize the other two stars in Eq.~\eqref{eq:K4stars}. Thus the three conjugate row pairs reproduce the star decomposition before any matrix multiplication is performed. Since $(\Rconj)_{\mu,7-\mu}=1$ for this ordering,
\begin{equation}
Q_4^{\mathsf T}\Rconj Q_4=J_4-I_4,
\label{eq:Q4check}
\end{equation}
where $J_4$ is the all-ones matrix and $I_4$ is the identity; their difference is the adjacency matrix of $K_4$. The star decomposition groups the requested links into spectral layers, while $Q_4$ records the delivered wavelength channels. Both describe the same physical assignment.

The same three-layer structure, comprising six wavelength channels forming three conjugate pairs distributed among four users, underlies the entanglement-based segment of the Slovak Quantum Communication Infrastructure, in which a single broadband source distributes three conjugate wavelength pairs over deployed telecommunication fiber to four city nodes. 
The pair-connection count $C_{uv}(Q)$, and with it the
uniquely connected condition for complete networks, were introduced
for the deployment setting in Refs.~\cite{AktasCAMA2025,companionB}.
Here we generalize this formulation to an arbitrary requested graph
$G$, where $C(Q)$ records the multiplicity of each requested connection.

{Ref.~\cite{companionB} also reports the performance of the
deployed network and shows how the connectivity-matrix framework was used to
analyse and optimize its wavelength assignment.


We use the matrix certificate to check complete assignments and the graph representation to optimize them. The standard local biclique parameters are~\cite{Pinto2014Biclique}
\begin{align}
\operatorname{lbp}(G)
&=
\min_{\mathcal P\text{ a biclique partition of }G}
\chi(\mathcal P),
\label{eq:lbp}\\
\operatorname{lbc}(G)
&=
\min_{\mathcal P\text{ a biclique cover of }G}
\chi(\mathcal P).
\label{eq:lbc}
\end{align}
The first forbids repeated edges; the second permits them. Thus, the receiver wavelength load introduced above is exactly the local biclique incidence studied in graph theory.

The graph description applies to both one-sided and two-sided fan-out. We first turn to the one-sided case, in which one conjugate wavelength channel is delivered to a single center user and the other is distributed among a set of leaves.

\section{One-sided architectures}
\label{sec:one-sided}

We first fix the optical meaning of a one-sided layer. In each spectral layer, the user receiving the unsplit wavelength channel is the \emph{center}, and the users receiving its conjugate form the \emph{leaf set}. If one layer may reach at most $r$ leaves, then
\begin{equation}
1\le |B_\ell|\le r,
\end{equation}
and the layer realizes the star $S(c_\ell,B_\ell)$ with at most $r$ requested edges.

A nonredundant one-sided assignment is an edge-disjoint star decomposition,
\begin{equation}
E(G)=\mathop{\dot\bigcup}_{\ell=1}^{L}S(c_\ell,B_\ell).
\label{eq:star-decomp}
\end{equation}
Let $L_{\min}(G,r)$ be the minimum number of nonempty star layers with at most $r$ edges per star. For a star decomposition $\Dset$, define
\begin{align}
\rho(\Dset)&=\max_{S\in\Dset}|E(S)|,\\
\rho_G(L)&=
\min_{\substack{\Dset\text{ a star decomposition of }G\\|\Dset|=L}}
\rho(\Dset).
\label{eq:rho-definition}
\end{align}
Thus, $\rho_G(L)$ is the smallest splitter fan-out compatible with exactly $L$ nonempty layers.

We now choose one endpoint of each requested edge as its center and orient the edge away from that endpoint. For an orientation $D$ of $G$, let $\delta_D^+(v)$ be the set of edges directed out of $v$ and let $d_D^+(v)=|\delta_D^+(v)|$. The edges in $\delta_D^+(v)$ are precisely those assigned to star layers centered at $v$. For a finite set $X$, let $\Pi_r(X)$ denote the set of partitions of $X$ into nonempty blocks of size at most $r$, with $\Pi_r(\varnothing)$ containing only the empty partition.

A choice $\mathcal P_v\in\Pi_r(\delta_D^+(v))$ specifies the layers centered at $v$: each block $S\in\mathcal P_v$ forms one star. Let $c(v,S)\ge0$ be the additive cost of using this block as a layer. Depending on the model, it may count one occupied spectral layer, the required pair-generation rate, splitter loss, or another engineering resource.

The orientation records which user is responsible for each requested edge, while the block partitions record how those edges are placed into physical layers. We can therefore optimize the complete one-sided assignment by optimizing these two choices. The orientation description is a standard tool for star decompositions~\cite{Harangi2025Orientations}; here the block cost makes its optical resource interpretation explicit. 

Let $\Oset(G)$ denote the set of all orientations of $G$. 

\begin{theorem}[Orientation-block characterization]
\label{thm:master-one-sided}
For every nonnegative block cost $c$, the minimum total cost over all nonredundant one-sided assignments with fan-out at most $r$ is
\begin{equation}
C_{G,r}^{\star}[c]
=
\min_{D\in\Oset(G)}
\sum_{v\in V}
\min_{\mathcal P_v\in\Pi_r(\delta_D^+(v))}
\sum_{S\in\mathcal P_v}c(v,S).
\label{eq:master-one-sided}
\end{equation}
\end{theorem}

\begin{proof} A nonredundant star decomposition assigns every edge to one center.
 Orienting the edge away from this center gives an orientation $D$. For each vertex $v$, the stars centered at $v$ partition $\delta_D^+(v)$ into blocks of size at most $r$, and their additive cost is the expression on the right-hand side.

Conversely, an orientation and one block partition at each vertex define a star for every block. Every edge has exactly one tail and therefore belongs to exactly one star. The resulting assignment is nonredundant and has the same cost.

This is an exact finite optimization, not a claim of an efficient algorithm for arbitrary graphs. In particular, its unbounded-fan-out specialization below contains minimum vertex cover.
\end{proof}


Setting $c(v,S)=1$ recovers the spectral layer problem. 
For a fixed orientation, the smallest number of blocks at $v$ is $\lceil d_D^+(v)/r\rceil$. Hence
\begin{equation}
\Phi_G(r)
=
\min_{D\in\Oset(G)}
\sum_{v\in V}
\left\lceil\frac{d_D^+(v)}{r}\right\rceil.
\label{eq:Phi}
\end{equation}

The minimum number of nonempty star layers with fan-out at most $r$ is
\begin{equation}
L_{\min}(G,r)=\Phi_G(r).
\label{eq:Lorientation}
\end{equation}
Let $\tau(G)$ be the minimum vertex-cover size of $G$.
For every feasible layer count $\tau(G)\le L\le m$, the smallest possible maximum star size is
\begin{equation}
\rho_G(L)
=
\min\{r\in\mathbb N:\Phi_G(r)\le L\}.
\label{eq:rhoorientation}
\end{equation}

Equation~\eqref{eq:Lorientation} fixes the smallest layer count for a given fan-out. Equation~\eqref{eq:rhoorientation} reads the same construction in the opposite direction: it selects the smallest fan-out for which no more than $L$ layers are needed. If the resulting decomposition has fewer than $L$ layers, splitting a nontrivial star increases the count one layer at a time without repeating an edge.


Without a fan-out bound, all edges assigned to the same center may be placed in one star, and the minimum number of layers is determined by the smallest set of centers meeting every requested edge. Thus, 

Let $\tau(G)$ be the minimum vertex-cover size of $G$. Then
\begin{equation}
L_{\min}(G,\infty)=\tau(G).
\label{eq:tau}
\end{equation}
Moreover, a decomposition into exactly $L$ nonempty stars exists if and only if
\begin{equation}
\tau(G)\le L\le m.
\label{eq:Lrange}
\end{equation}

 The centers of a star decomposition meet every edge and therefore form a vertex cover. Conversely, choose a minimum vertex cover $C$. Since $C$ is inclusion minimal, every $c\in C$ has a neighbor outside $C$. Assign one such edge to each $c$, and assign every remaining edge to either endpoint that lies in $C$. The assigned edges form $|C|=\tau(G)$ nonempty stars. Splitting a star with at least two edges increases the layer count by one, giving every value of $L$ up to the edge-by-edge decomposition with $m$ layers.

Before specializing the graph, we record two lower bounds that do not require an orientation. The star centers form a vertex cover, and $L$ stars of fan-out at most $r$ contain at most $Lr$ edges. Therefore,
\begin{equation}
L_{\min}(G,r)
\ge
\max\left\{\tau(G),\left\lceil\frac mr\right\rceil\right\},
\qquad
\rho_G(L)\ge\left\lceil\frac mL\right\rceil.
\label{eq:basic-one-sided-bounds}
\end{equation}
The first obstruction counts the centers that must exist; the second counts the edge positions available across all layers.

\subsection{Balanced profiles in dense graphs}
\label{subsec:balanced-profiles}


A balanced star-size profile attains the counting bound whenever it is realizable. We describe that profile by writing
\begin{equation}
m=qL+s,
\qquad 0\le s<L.
\label{eq:balanced-profile-general}
\end{equation}
The balanced profile contains $L-s$ stars of size $q$ and $s$ stars of size $q+1$. Its largest star has size
\begin{equation}
r_L=\left\lceil\frac mL\right\rceil,
\label{eq:rL-definition}
\end{equation}
which is the smallest value allowed by the counting bound. The remaining question is whether the edges can be grouped into stars with these prescribed sizes. Tarsi's theorem gives a sufficient condition~\cite{Tarsi1981Stars}.

\begin{lemma}[Balanced star partition under a minimum-degree condition]
\label{thm:dense}
For a feasible layer count $L$, if
\begin{equation}
\delta(G)\ge\frac n2+r_L-1,
\label{eq:dense-cond}
\end{equation}
then $G$ admits the balanced profile in Eq.~\eqref{eq:balanced-profile-general}, and
\begin{equation}
\rho_G(L)=\left\lceil\frac mL\right\rceil.
\label{eq:dense-law}
\end{equation}
\end{lemma}

We use the lemma through its minimum-degree condition. Appendix~\ref{app:dense-proof} applies Tarsi's minimum-degree prescribed-star theorem directly to the balanced size list.

We now apply these results to fully connected networks. This case is important experimentally because every pair of users must share an entangled link, and the symmetry of the complete graph allows the layer--fan-out tradeoff to be determined in closed form.

\subsection{Complete and one-excluded-partner networks}
\label{subsec:complete-networks}

For $N$ users, let
\begin{equation}
E_N=\binom N2=\frac{N(N-1)}2
\label{eq:EN}
\end{equation}
be the number of requested links in $K_N$. Since $\tau(K_N)=N-1$, at least $N-1$ star layers are required. The capacity bound gives the second restriction $L\ge\lceil E_N/r\rceil$. For complete graphs, these two obstructions are sufficient.

\begin{theorem}[Exact layer--fan-out tradeoff for $K_N$]
\label{thm:complete}
For $N-1\le L\le E_N$,
\begin{equation}
\rho_{K_N}(L)=\left\lceil\frac{E_N}{L}\right\rceil.
\label{eq:complete-rho}
\end{equation}
Equivalently, for every integer $r\ge1$,
\begin{equation}
L_{\min}(K_N,r)
=
\max\left\{N-1,\left\lceil\frac{E_N}{r}\right\rceil\right\}.
\label{eq:complete-L}
\end{equation}
\end{theorem}

\begin{proof}
The lower bounds follow from $\tau(K_N)=N-1$ and $|E(K_N)|=E_N$.  At $L=N-1$, choose one user $z$ that is never a center. Orient the complete graph on the remaining $N-1$ users so that the out-degrees differ by at most one. If $N-1$ is odd, use a regular tournament, so every out-degree is $(N-2)/2$. If $N-1$ is even, delete one vertex from a regular tournament on $N$ vertices; the maximum out-degree is then $(N-1)/2$.

For each remaining user $v$, form one star whose leaves are $z$ and the out-neighbors of $v$. Every edge appears once, and the largest star has size
\begin{equation}
1+\max_v d_D^+(v)
=
\left\lceil\frac N2\right\rceil
=
\left\lceil\frac{E_N}{N-1}\right\rceil.
\label{eq:min-complete-profile}
\end{equation}
For $L\ge N$, put $r_L=\lceil E_N/L\rceil$. Then $r_L\le\lceil(N-1)/2\rceil$ and
\begin{equation}
\delta(K_N)=N-1\ge\frac N2+r_L-1.
\end{equation}
Lemma~\ref{thm:dense} therefore realizes the balanced profile. This proves Eq.~\eqref{eq:complete-rho}; inversion gives Eq.~\eqref{eq:complete-L}.
\end{proof}

For every $N-1\le L\le E_N$, write $E_N=qL+s$ with $0\le s<L$. Then $K_N$ has a nonredundant one-sided decomposition into $L-s$ stars of size $q$ and $s$ stars of size $q+1$.

We record the full fan-out profile, not only the size of the largest star, because we will use it when layers of different sizes incur different optical loss. Figure~\ref{fig:layercount} illustrates Eq.~\eqref{thm:complete} for $K_{20}$.

The same law remains tractable when each user has exactly one excluded partner. For even $N$, remove a perfect matching $M$ and let
\begin{equation}
CP_N=K_N-M
\end{equation}
be the cocktail-party graph, obtained by removing one perfect matching so that every user has one excluded partner. Its resources satisfy
\begin{equation}
|E(CP_N)|=\frac{N(N-2)}2,
\qquad
\tau(CP_N)=N-2.
\label{eq:cp-range}
\end{equation}

\begin{corollary}[Exact layer--fan-out tradeoff with one excluded partner per user]
\label{thm:cocktail}
For even $N\ge4$ and $N-2\le L\le N(N-2)/2$,
\begin{equation}
\rho_{CP_N}(L)
=
\left\lceil\frac{N(N-2)}{2L}\right\rceil.
\label{eq:cp-rho}
\end{equation}
Equivalently,
\begin{equation}
L_{\min}(CP_N,r)
=
\max\left\{N-2,\left\lceil\frac{N(N-2)}{2r}\right\rceil\right\}.
\label{eq:cp-L}
\end{equation}
\end{corollary}

We construct the optimum by using the two endpoints of a deleted matching edge as noncenters and balancing the orientation on the remaining users. We give the details in Appendix~\ref{app:cocktail-proof}. Figure~\ref{fig:layercount} illustrates the exact one-sided layer--fan-out law for $K_{20}$.

\begin{figure}[t]
\centering
\includegraphics[width=0.96\columnwidth]{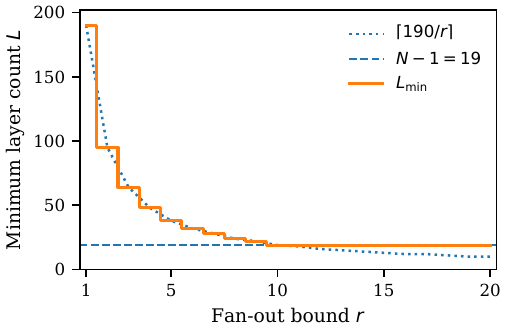}
\caption{Minimum number of one-sided spectral layers for $K_{20}$ as a function of the fan-out bound $r$. The dotted curve is the edge-capacity bound $\lceil190/r\rceil$, the dashed line is the vertex-cover bound $N-1=19$, and the solid curve is the exact value $L_{\min}(K_{20},r)$ from Eq.~\eqref{thm:complete}. The edge-capacity bound determines the optimum for $r<10$, while the vertex-cover bound fixes $L_{\min}=19$ for $r\ge10$.}
\label{fig:layercount}
\end{figure}
Allowing the second conjugate wavelength channel to split replaces stars by general bicliques. We must then optimize the total number of layers and the receiver wavelength load separately.

\section{Two-sided fan-out}
\label{sec:two-sided}

We begin by distinguishing two resources. A biclique partition uses every requested edge exactly once, while a biclique cover may repeat edges. Let $\operatorname{bp}(G)$ and $\operatorname{bc}(G)$ denote the minimum numbers of bicliques in a partition and cover, respectively. The local quantities $\operatorname{lbp}(G)$ and $\operatorname{lbc}(G)$ were defined in Eqs.~\eqref{eq:lbp} and~\eqref{eq:lbc}.

For a complete graph, the minimum partition layer count grows linearly, whereas the minimum receiver load grows logarithmically:For every $N\ge2$,

\begin{align}
\operatorname{bp}(K_N)&=N-1,
\label{eq:complete-bp}\\
\operatorname{lbc}(K_N)
=
\operatorname{lbp}(K_N)
&=
\left\lceil\log_2N\right\rceil.
\label{eq:complete-local-biclique}
\end{align}

The first identity is the Graham--Pollak theorem~\cite{GrahamPollak1971}; the local bound follows from the Katona--Szemer\'edi theorem and its biclique formulation~\cite{KatonaSzemeredi1967,Pinto2014Biclique}. We include the two short lower-bound arguments because they have direct optical meanings.


For the layer-count bound, suppose $K_N$ were partitioned into $L\le N-2$ bicliques $\Bic(A_\ell,B_\ell)$. Choose a nonzero vector $x\in\mathbb R^N$ satisfying
\begin{equation}
\sum_{i=1}^{N}x_i=0,
\qquad
\sum_{i\in A_\ell}x_i=0
\quad(1\le\ell\le L).
\label{eq:gp-linear-system}
\end{equation}
Such a vector exists because there are at most $N-1$ homogeneous equations. Since every pair lies across exactly one biclique,
\begin{equation}
\sum_{i<j}x_ix_j
=
\sum_{\ell=1}^{L}
\left(\sum_{i\in A_\ell}x_i\right)
\left(\sum_{j\in B_\ell}x_j\right)
=0.
\label{eq:gp-edge-identity}
\end{equation}
Together with $\sum_i x_i=0$, this gives $\sum_i x_i^2=0$, contradicting the choice of $x$. Hence $L\ge N-1$.

For the receiver-load lower bound,
independently choose one side of every biclique in a cover, each with probability $1/2$. Call a user compatible when the chosen side agrees with that user in every biclique containing it. Two distinct users cannot both be compatible, because some covering biclique places them on opposite sides. Hence the expected number of compatible users is at most one. If every user occurs in at most $k$ bicliques, then
\begin{equation}
1\ge\sum_{u\in V}2^{-\chi_u}\ge N2^{-k},
\label{eq:local-biclique-probability}
\end{equation}
so $k\ge\lceil\log_2N\rceil$. The binary-tree construction below attains both lower bounds at once.

\subsection{Hierarchical biclique architecture}
\label{subsec:binary-tree}

We now construct a partition that attains both lower bounds. Let $T$ be a rooted full binary tree whose leaves are the users. Every internal node $t$ divides the leaves below it into the leaf sets $A_t$ and $B_t$ of its two child subtrees. Assign one spectral layer realizing $\Bic(A_t,B_t)$.

\begin{theorem}[Hierarchical biclique partition from a binary tree]
\label{thm:tree}
The construction above satisfies
\begin{equation}
L=N-1,
\qquad
\Delta=0,
\qquad
\chi=\operatorname{height}(T).
\label{eq:tree-resources}
\end{equation}
A balanced tree has height $\lceil\log_2N\rceil$ and therefore attains both $\operatorname{bp}(K_N)$ and $\operatorname{lbp}(K_N)$.
\end{theorem}

\begin{proof}

A full binary tree with $N$ leaves has $N-1$ internal nodes, hence $N-1$ layers. 
Fix two users $u$ and $v$. Their lowest common ancestor $t$ is the unique internal node whose two child subtrees separate them, so $uv$ lies in $\Bic(A_t,B_t)$. At every proper descendant of $t$, at most one of the two users is present; at every ancestor, they lie in the same child subtree. The edge therefore occurs in exactly one layer.

A user belongs to the biclique associated with each internal node on its leaf-to-root path and to no other layer. Its incidence is its leaf depth, and the maximum incidence is the tree height. A recursively balanced split has height $\lceil\log_2N\rceil$, which meets the lower bound in Eq.~\eqref{eq:complete-local-biclique}.
\end{proof}

We can read the same construction directly as an optical architecture. At each internal node, one conjugate wavelength is distributed to the users in one child subtree and the other to the users in the second. Two users are connected at the first node where their paths separate. The hierarchy therefore keeps the minimum number of spectral layers among nonredundant assignments while reducing the receiver load from linear to logarithmic. The $N-1$ layer optimum is a partition statement; covers that repeat edges may use fewer layers. Indeed, the unrestricted cover optimum is $\bc(K_N)=\lceil\log_2N\rceil$~\cite{Pinto2014Biclique}, so Theorem~\ref{thm:tree} must not be interpreted as minimizing the number of layers over redundant covers. Its cost is larger two-sided fan-out near the root.

\subsection{Bounded side size}
\label{subsec:bounded-two-sided}

We now impose the hardware constraint that each wavelength may reach at most $r$ users, so every layer satisfies $|A|,|B|\le r$. A user meets at most $r$ new partners in one incident layer. Since each user of $K_N$ has $N-1$ partners,
\begin{equation}
\chi_u\ge\left\lceil\frac{N-1}{r}\right\rceil,
\qquad
\chi\ge\left\lceil\frac{N-1}{r}\right\rceil.
\label{eq:biclique-load-bound}
\end{equation}
One layer also contains at most $r^2$ edges.

Any cover of $K_N$ by bicliques with $|A|,|B|\le r$ obeys
\begin{equation}
L\ge
\max\left\{
\left\lceil\frac{N(N-1)}{2r^2}\right\rceil,
\left\lceil
\frac{N\left\lceil(N-1)/r\right\rceil}{2r}
\right\rceil
\right\}.
\label{eq:biclique-cover-bound}
\end{equation}

We obtain the two bounds by counting edge positions and user--layer incidences. The first term counts edge positions: $K_N$ has $N(N-1)/2$ edges and one bounded biclique has at most $r^2$. For the second, summing Eq.~\eqref{eq:biclique-load-bound} over all users gives at least $N\lceil(N-1)/r\rceil$ user--layer incidences, while one layer contains at most $2r$ users.

A block hierarchy attains both lower bounds asymptotically. when $r$ is fixed and $N\to\infty$. Divide the users into groups of size at most $r$, connect every pair of groups by one biclique, and use the binary-tree construction inside each group.

\begin{proposition}[Block-hierarchy construction for bounded two-sided fan-out]
\label{thm:block-hierarchy}
Let $1\le r\le N$ and $g=\lceil N/r\rceil$. Then there exists a nonredundant biclique partition of $K_N$ with side sizes at most $r$ such that
\begin{align}
L&=\binom g2+N-g,
\label{eq:block-hierarchy-L}\\
\chi&\le g-1+\left\lceil\log_2r\right\rceil.
\label{eq:block-hierarchy-chi}
\end{align}

For fixed $r$ and $N\to\infty$,
\begin{equation}
L=\frac{N^2}{2r^2}+O(N),
\qquad
\chi=\frac Nr+O(1).
\label{eq:block-hierarchy-asymptotic}
\end{equation}
\end{proposition}

We compute the layer count and receiver-load bound from this construction in Appendix~\ref{app:block-proof}.

We recover the two limiting architectures directly. At $r=1$, the construction becomes the edge-by-edge partition. At $r=N$, all users lie in one block and the construction reduces to the binary-tree partition. Between these limits, $r$ controls the tradeoff between layer count, receiver load, and two-sided fan-out.


For $K_8$ the tradeoff is fully explicit: every layer can be listed and every repetition and receiver incidence checked edge by edge.

\section{Constrained architecture tradeoff for the eight-user mesh}
\label{sec:k8-frontier}

We now compare five passive architectures for the complete graph $K_8$, which has $28$ requested links. The architectures below separate the number of spectral layers $L$, the edge-repetition overhead $\Delta$, the maximum receiver load $\chi$, and the largest side size used in any layer.

\begin{proposition}[Exact $K_8$ resource points under fan-out constraints]
\label{thm:k8-frontier}
\label{thm:k8-trichotomy}
For $K_8$, the following statements hold.
\begin{enumerate}
\item Every nonredundant passive assignment requires at least seven layers. A balanced-tree biclique partition attains $L=7$, $\Delta=0$, and the globally minimum receiver load $\chi=3$.
\item A seven-layer one-sided partition exists with maximum star size four. Every seven-layer star partition has $\chi=7$.
\item Under $|A|,|B|\le2$, a cover requires at least eight layers. Every eight-layer cover has $\Delta=4$ and $\chi=4$, and such a cover exists.
\item Under the same side-size bound, a nonredundant assignment requires at least nine layers. Nine layers are achievable with $\chi=4$.
\item  Under $|A|,|B|\le1$, every minimum-layer assignment is the edge-by-edge partition, up to layer ordering and interchange of the two conjugate wavelength sides, with $L=28$, $\Delta=0$, and $\chi=7$. 
\end{enumerate}
\end{proposition}

In Appendix~\ref{app:k8-proof}, we prove the lower bounds and verify the explicit decompositions for all five resource points. The five resource points are summarized in Table~\ref{tab:k8-topology}.

\begin{table*}[t]
\caption{Constrained resource points for the complete eight-user network. The layer type gives the largest two-sided fan-out occurring in each architecture.}
\label{tab:k8-topology}
\begin{ruledtabular}
\begin{tabular}{lccccl}
architecture & $L$ & largest layer type & $\Delta$ & $\chi$ & defining constraint \\
\hline
seven-star partition & $7$ & $1\times4$ & $0$ & $7$ & one-sided fan-out \\
seven-layer hierarchy & $7$ & $4\times4$ & $0$ & $3$ & minimum receiver load \\
eight-layer cover & $8$ & $2\times2$ & $4$ & $4$ & minimum $L$ with side size at most two \\
nine-layer partition & $9$ & $2\times2$ & $0$ & $4$ & nonredundant, side size at most two \\
pairwise unsplit & $28$ & $1\times1$ & $0$ & $7$ & no-fan-out endpoint \\
\end{tabular}
\end{ruledtabular}
\end{table*}
These entries correspond to different hardware restrictions, not to a universal ranking. The hierarchy minimizes $L$ and $\chi$ among nonredundant assignments but requires a $4\times4$ root layer. The star partition keeps one wavelength unsplit, the side-size-at-most-two designs limit both splitter sides, and the pairwise assignment removes fan-out altogether. We can rank the architectures physically only after specifying the splitter-loss and source models.

\section{Pair flux and pairwise key rates}
\label{sec:qkdmodel}

The graph determines which assignments are feasible. We now determine the pair-generation rate required by each assignment, first for independently adjustable spectral layers (Secs.~VI--VIII) and then for a common broadband source (Secs.~IX--X).

We first consider an independent-layer source model. It represents separate source elements, independently pumped spectral channels, or another implementation in which the pair-generation rate supplied to each occupied layer can be adjusted separately. It is distinct from the single broadband-source model introduced later.

\begin{definition}[Separable edge-response assumptions]
\label{def:separable-model}
For every occupied star $S$, let $\Sset_S\subset[0,1]^S$ be a nonempty compact set of physically realizable useful branch-probability vectors. We assume that
\begin{enumerate}
\item the source model supplies each occupied layer with an independently adjustable nonnegative pair-generation rate;
\item the branch vector of layer $S$ may be chosen only from $\Sset_S$;
\item the wavelength-layer label is retained whenever an edge occurs in more than one layer;
\item each edge has an asymptotic stationary response determined only by the pair flux delivered to that edge; and
\item there is no detector saturation, cross-edge accidental coupling, or receiver constraint not already included in that calibrated response.
\end{enumerate}
For every edge $e$, let $F_e:[0,\infty)\to[0,\infty)$ denote its exact nondecreasing continuous asymptotic response under the declared model. If only a certified achievable lower bound $\underline F_e$ is available, replacing $F_e$ by $\underline F_e$ gives an achievable inner region rather than an exact rate region.
\end{definition}

For secret-key rates, the responses must be simultaneously valid for the specified processing rule. An edgewise formula alone does not certify joint security when raw events are reused; the additional requirement is stated in Sec.~\ref{sec:cw}.

The continuous-wave model in Sec.~\ref{sec:cw} is generally a joint layer response because one detector channel can participate in several coincidence measurements within a fan-out layer. It reduces to Definition~\ref{def:separable-model} only on an operating family for which each edge response depends on that edge's delivered flux alone, for example when accidental coincidences are negligible. Calibrating a one-dimensional slice does not establish this property for other allocations.

We quantify the flux required for a target rate by defining the least supporting pre-protocol pair flux
\begin{equation}
\phi_e(R)=F_e^{\leftarrow}(R)
=\inf\{x\ge0:F_e(x)\ge R\},
\label{eq:inverse-response}
\end{equation}
with $\phi_e(R)=+\infty$ outside the response range.

For a star $S$, the target vector $\bm R=(R_e)_{e\in S}$ fixes a required flux $\phi_e(R_e)$ on every edge. If the layer operates at pair-generation rate $B_S$ with branch vector $\bm s$, edge $e$ receives flux $B_Ss_e$. The layer rate must therefore satisfy $B_S\ge\phi_e(R_e)/s_e$ for every edge. Minimizing the largest of these requirements over the realizable splitter set gives
\begin{equation}
\Psi_S(\bm R)
=
\min_{\bm s\in\Sset_S}
\max_{e\in S}
\frac{\phi_e(R_e)}{s_e},
\label{eq:general-star-cost}
\end{equation}
where a positive numerator divided by zero is $+\infty$ and $0/0$ is taken as zero. Thus, $\Psi_S(\bm R)$ is the least input pair-generation rate with which the physical splitter can satisfy all edge targets in that star.

\begin{proposition}[Exact rate region for a fixed star decomposition]
\label{thm:protocol-region}
Let $\Dset$ be a nonredundant star decomposition. Under Definition~\ref{def:separable-model}, its simultaneous rate region under total pair-rate budget $B_{\rm tot}$ is
\begin{equation}
\Rset_{\Dset}(B_{\rm tot})
=
\left\{\bm R\ge0:
\sum_{S\in\Dset}\Psi_S(\bm R)
\le B_{\rm tot}
\right\}.
\label{eq:general-region}
\end{equation}
\end{proposition}

\begin{proof}
For a chosen branch  vector $\bm s\in\Sset_S$, every target rate in star $S$ requires
\begin{equation}
B_Ss_e\ge\phi_e(R_e)
\qquad(e\in S).
\end{equation}
Hence $B_S\ge\max_e\phi_e(R_e)/s_e$, and minimizing over the realizable branch vectors gives the necessary bound $B_S\ge\Psi_S(\bm R)$. Summing over the occupied stars proves necessity. Compactness of $\Sset_S$ and lower semicontinuity of the objective ensure that the minimum in Eq.~\eqref{eq:general-star-cost} is attained. For every finite demand within the response range, continuity and monotonicity give $F_e(\phi_e(R_e))\ge R_e$. Choosing one minimizing branch vector in each star and setting $B_S=\Psi_S(\bm R)$ therefore proves sufficiency.
\end{proof}

Two common splitter models follow immediately. If the branch ratios are continuously tunable over the useful simplex
\begin{equation}
\Sset_S
=
\left\{\bm s\ge0:
\sum_{e\in S}s_e\le\eta_S^{\rm split}
\right\},
\label{eq:tunable-simplex}
\end{equation}
then proportional allocation gives
\begin{equation}
\Psi_S(\bm R)
=
\frac{1}{\eta_S^{\rm split}}
\sum_{e\in S}\phi_e(R_e).
\label{eq:tunable-star-cost}
\end{equation}
For a fixed branch vector $\bm s^{(0)}$, one instead has
\begin{equation}
\Psi_S(\bm R)
=
\max_{e\in S}
\frac{\phi_e(R_e)}{s_e^{(0)}}.
\label{eq:fixed-star-cost}
\end{equation}
The sum formula is therefore exact only for the tunable-simplex model; a fixed splitter is governed by the maximum-ratio expression.

If $\eta_S^{\rm split}=\eta^{\rm split}$ is the same for every layer and the tunable-simplex model applies, the rate region is independent of the star decomposition. Every requested edge contributes once to the total source budget. Equation~\eqref{eq:general-region} then reduces to
\begin{equation}
\Rset(B_{\rm tot})
=
\left\{\bm R\ge0:
\sum_{e\in E(G)}\phi_e(R_e)
\le B_{\rm tot}\eta^{\rm split}
\right\}.
\label{eq:invariant-region}
\end{equation}
Thus, fan-out-independent loss makes the edge-separable rate region independent of the particular star decomposition. The decomposition becomes physically relevant when the useful transmission depends on fan-out, when the branch ratios are fixed, or when the source fixes the relative layer rates.

\subsection{Label-resolved Bell-pair model}
\label{subsec:flagged-model}

A useful normalization is the following labelled Bell-pair model. The model retains the successful edge label and processes each labelled stream only for the two users at the endpoints of that edge. We call this \emph{direct edgewise processing}. It excludes key relaying, joint conversion of entanglement across different edge labels, and multipartite distillation.

For one emission in layer $\ell$, let $p_{\ell e}$ be the probability that edge $e$ is selected and accepted. Conditioned on this event, the endpoints of $e$ receive an ideal Bell pair. The output state is
\begin{equation}
\omega_\ell
=
\sum_{e\in S_\ell}
 p_{\ell e}|e\rangle\!\langle e|_F\otimes\Phi_e
+\omega_\ell^{\rm er},
\label{eq:flaggedstate}
\end{equation}
where $F$ stores the edge label and $\omega_\ell^{\rm er}$ contains the unsuccessful events and is assumed to contain no distillable key between any user pair. These erasure outcomes are discarded. 

Under direct edgewise processing, each labelled block is used only to generate a direct key between its own endpoints, without key relaying or joint network distillation. The asymptotic secret-key rates per source emission are
\begin{equation}
\Kset_\ell
=
\{\bm K_\ell:0\le K_{\ell e}\le p_{\ell e}
\text{ for every }e\in S_\ell\}.
\label{eq:flagged-region}
\end{equation}
Every rate vector in this box is attainable.

 After $n$ emissions, the block labelled by edge $e$ contains asymptotically $np_{\ell e}$ Bell pairs. Processing the labelled blocks separately attains every rate vector in Eq.~\eqref{eq:flagged-region}. Conversely, direct edgewise processing assigns the key of edge $e$ only to the block carrying label $e$. One Bell pair supplies at most one secret bit under this processing rule, so that block cannot yield more than $np_{\ell e}$ secret bits asymptotically. Dividing by $n$ gives $K_{\ell e}\le p_{\ell e}$.

This fixes the normalization used below: one accepted, publicly labelled Bell pair supplies one unit of direct pairwise key. It is not a statement about the unrestricted network key capacity, for which relaying or joint processing across different edge labels may change the achievable region~\cite{Takeoka2017Broadcast}.

\subsection{Linear responses and service constraints}

If the edge error parameters do not change with allocated flux, the linear response is
\begin{equation}
F_e(x)=\gamma_ex,
\qquad \gamma_e>0.
\label{eq:linear-response}
\end{equation}
The coefficient $\gamma_e$ is the number of secret bits obtained per unit pre-protocol pair flux. For asymptotic BBM92~\cite{BBM1992} with flux-independent bit and phase errors,
\begin{equation}
\gamma_e=q_eg_e
\left[
1-f_{{\rm EC},e}H_2(E_e^b)-H_2(E_e^p)
\right]_+.
\label{eq:gammaBBM}
\end{equation}
Here $q_e$ is the basis-sifting probability, $g_e$ is the accepted-pair probability after path loss and timing selection, and $f_{{\rm EC},e}$ is the error-correction inefficiency. We use $H_2(x)=-x\log_2x-(1-x)\log_2(1-x)$ for the binary entropy and $[x]_+=\max\{x,0\}$. Passive-basis and threshold-detector implementations require the corresponding receiver security proof~\cite{TsurumaruTamaki2008,Kawakami2026PassiveQKD}. If the positive-part bracket vanishes, then $\gamma_e=0$ and a positive target on edge $e$ is infeasible in this linear model. All formulas below that use $\gamma_e^{-1}$ are therefore restricted to edges with $\gamma_e>0$.  For a linear response with $\gamma_e>0$, $\phi_e(R)=R/\gamma_e$.

Under the common tunable-simplex model, Eq.~\eqref{eq:invariant-region} becomes
\begin{equation}
\Rset_{\rm lin}(B_{\rm tot})
=
\left\{\bm R\ge0:
\sum_{e\in E(G)}\frac{R_e}{\gamma_e}
\le B_{\rm tot}\eta^{\rm split}
\right\}.
\label{eq:linear-simplex}
\end{equation}
The largest common edge rate is therefore
\begin{equation}
R_{\rm mm}^{\star}
=
\frac{B_{\rm tot}\eta^{\rm split}}
{\sum_e\gamma_e^{-1}}.
\label{eq:linear-maxmin}
\end{equation}

We now impose a common pair-generation rate $B$ on every occupied layer. Here the total occupied rate is $LB$, rather than a fixed $B_{\rm tot}$.  Under this constraint, a star $S$ supports common edge rate $R$ only if
\begin{equation}
R\sum_{e\in S}\gamma_e^{-1}
\le B\eta^{\rm split}.
\end{equation}
Define its weighted load by
\begin{equation}
W(S)=\sum_{e\in S}\gamma_e^{-1}
\label{eq:weighted-star-load}
\end{equation}
and, for exactly $L$ layers,
\begin{equation}
\Omega_G(L;\bm\gamma)
=
\min_{\substack{\Dset\text{ a star decomposition}\\|\Dset|=L}}
\max_{S\in\Dset}W(S).
\label{eq:Omega}
\end{equation}
The most demanding star limits the common rate:
\begin{equation}
R_{\rm mm}^{\star}(G,L|B)
=
\frac{B\eta^{\rm split}}
{\Omega_G(L;\bm\gamma)}.
\label{eq:equal-layer-rate}
\end{equation}
When all edges have the same gain $\gamma_e=\gamma$, this reduces to
\begin{equation}
R_{\rm mm}^{\star}(G,L|B)
=
\frac{B\gamma\eta^{\rm split}}{\rho_G(L)}.
\label{eq:hom-rate}
\end{equation}

We can also use the same source budget to enforce proportional service. If $d_e>0$ and $R_e\ge d_et$ is required, the largest common service level is
\begin{equation}
t^\star
=
\frac{B_{\rm tot}\eta^{\rm split}}
{\sum_e d_e/\gamma_e},
\qquad
R_e^\star=d_et^\star.
\label{eq:service}
\end{equation}
The weights $d_e$ fix the required proportions, while $1/\gamma_e$ is the pair flux needed for one unit of rate on edge $e$.

When larger stars incur greater loss, the full star-size profile enters the source budget, and the decomposition invariance no longer holds.

This decomposition invariance disappears once useful transmission depends on fan-out. Once larger stars suffer greater loss, the full star-size profile enters the source budget.

\section{Fan-out-dependent splitter loss}
\label{sec:hardware}

Fan-out changes the physical ordering because larger splitter trees incur additional loss as the number of leaves increases. We therefore retain the star decomposition and let the useful splitter transmission depend on the fan-out.

For a splitter serving $k$ leaves, let
\begin{equation}
\eta_{\rm s}(k)\in(0,1]
\label{eq:etak}
\end{equation}
be the measured probability that an input photon exits through one of the useful leaf branches, including insertion loss. If $s_e$ is the branch probability assigned to edge $e$ of a $k$-edge star, then
\begin{equation}
\sum_{e\in S}s_e\le\eta_{\rm s}(k).
\label{eq:fan-out-budget}
\end{equation}
A wavelength-dependent splitter is handled by replacing $\eta_{\rm s}(k)$ with $\eta_{{\rm s},\ell}(k)$.

The quantum network provider (QNP) in the companion deployment~\cite{companionB} provides an apparatus-level example. Figure~\ref{fig:deployment} shows the physical fiber routes from the source in Trnava to Nitra, Bratislava, and Tren\v{c}\'in; the reported link losses are $14.569$, $17.536$, and $24.703\,\mathrm{dB}$, respectively. The three fiber arms form a physical transport star, not the logical entanglement graph. The three-layer wavelength assignment in Eq.~\eqref{eq:K4stars} realizes the six requested links of $K_4$. The QNP divides the signal channels with $1\times2$ beam splitters. In the stage model below, $\eta_{\rm st}$ denotes useful transmission per splitter stage; fiber loss is included separately in the edge gain. An end-to-end calibration may be used instead, but the same loss must not also be charged stage by stage.

\begin{figure}[htbp]
\includegraphics[width=\columnwidth]{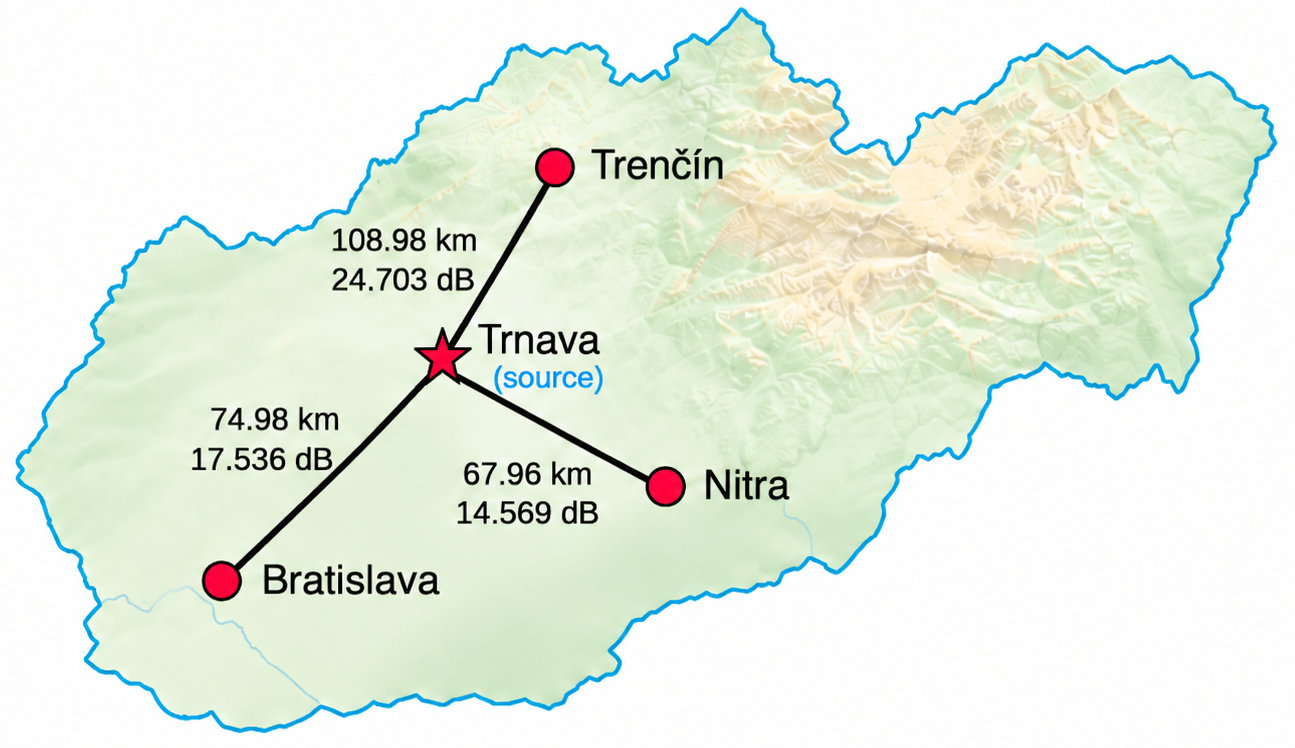}
\caption{Physical deployment of the four-node network described in the companion paper~\cite{companionB}. Trnava hosts the source/QNP; the labels give the reported fiber lengths and link losses to Bratislava, Nitra, and Tren\v{c}\'in. The physical routes form a three-arm star, whereas the wavelength assignment in Eq.~\eqref{eq:K4stars} realizes the logical six-link graph $K_4$. The map is a transport-layout illustration, not a graph-layer decomposition.}
\label{fig:deployment}
\end{figure}

\subsection{Source budget for a fixed star decomposition}

For a fixed star decomposition, the source requirement follows by summing the flux demanded behind each splitter. 
Suppose edge $e$ must produce rate $R_e$, and hence requires pre-protocol flux $\phi_e(R_e)$. All edge fluxes in one star pass through the same splitter. Their sum must therefore be divided by the useful transmission of that star.

\begin{proposition}[Rate region for tunable fan-out-dependent splitters]
\label{thm:hardware-region}
Assume that every $k$-leaf splitter is tunable over the simplex in Eq.~\eqref{eq:fan-out-budget}. Fix a star decomposition $\Dset$ and write $k_S=|E(S)|$. The rate region for separable edge responses is
\begin{equation}
\begin{aligned}
\Rset_{\Dset}^{\rm hw}(B_{\rm tot})
=\{\bm R\ge0:\;&
\sum_{S\in\Dset}
\frac{1}{\eta_{\rm s}(k_S)}\\[-1mm]
&\times\sum_{e\in S}\phi_e(R_e)
\le B_{\rm tot}\}.
\end{aligned}
\label{eq:hardware-region}
\end{equation}
For exactly $L$ layers, the architecture region is the union of these sets over all $L$-star decompositions of $G$.
\end{proposition}

We prove necessity by summing the edge flux required behind each splitter and prove sufficiency by proportional branch allocation. The complete proof is given in Appendix~\ref{app:hardware-proof}.

For linear responses and tunable branch ratios, the region simplifies as follows. One unit of common edge rate requires flux $1/\gamma_e$. The source pair-generation rate required by star $S$ is therefore
\begin{equation}
h(S)=
\frac{1}{\eta_{\rm s}(|E(S)|)}
\sum_{e\in S}\frac{1}{\gamma_e}.
\label{eq:star-hardware-cost}
\end{equation}
The sum is the useful output flux; division by $\eta_{\rm s}$ converts it to input pair-generation rate.

For exactly $L$ layers, define the minimum pair-generation rate per unit common key rate by
\begin{equation}
H_G(L)=
\min_{\substack{\Dset\text{ star decomposition of }G\\|\Dset|=L}}
\sum_{S\in\Dset}h(S).
\label{eq:HG}
\end{equation}
A common edge rate $R$ requires total pair-generation rate $RH_G(L)$.

For linear edge responses and tunable branch ratios,
\begin{equation}
R_{\rm mm}^{\star}(G,L)=
\frac{B_{\rm tot}}{H_G(L)}.
\label{eq:hardware-maxmin}
\end{equation}
If $\eta_{\rm s}(k)=\eta^{\rm split}$ is independent of $k$, then
\begin{equation}
H_G(L)=
\frac{1}{\eta^{\rm split}}
\sum_{e\in E(G)}\gamma_e^{-1},
\label{eq:constant-H}
\end{equation}
so the common rate is independent of the star decomposition.

\subsection{Complete meshes with identical edge gains}

For a complete mesh with identical edge gains, the optimization depends only on the star-size profile. Assume that every edge of $K_N$ has the same gain $\gamma$. A star with $k$ edges then has unit-rate cost $k/[\gamma\eta_{\rm s}(k)]$. Define
\begin{equation}
f(k)=\frac{k}{\eta_{\rm s}(k)}.
\label{eq:fk}
\end{equation}
The total number of assigned edges is fixed at $E_N$, so the optimal profile is governed by the increments of $f$. Assume that the sequence $f(1),\ldots,f(N-1)$ is discretely convex, meaning that
\begin{equation}
f(k+1)-f(k)
\ge f(k)-f(k-1),
\qquad 2\le k\le N-2.
\label{eq:discrete-convex}
\end{equation}
This is discrete convexity of the measured splitter-cost sequence.

Write
\begin{equation}
E_N=qL+s,
\qquad 0\le s<L.
\label{eq:division-profile}
\end{equation}
The balanced profile has $L-s$ stars of size $q$ and $s$ stars of size $q+1$.

\begin{theorem}[Minimum pair-generation requirement for $K_N$ under convex fan-out loss]
\label{thm:complete-hardware}
Suppose $f(k)=k/\eta_{\rm s}(k)$ is discretely convex. For $N-1\le L\le E_N$, the balanced profile in Eq.~\eqref{eq:division-profile} minimizes the total pair-generation requirement, and
\begin{equation}
H_{K_N}(L)=
\frac{1}{\gamma}
\left[
(L-s)\frac{q}{\eta_{\rm s}(q)}
+s\frac{q+1}{\eta_{\rm s}(q+1)}
\right].
\label{eq:complete-H}
\end{equation}
Consequently,
\begin{equation}
R_{\rm mm}^{\star}(K_N,L)=
\frac{B_{\rm tot}\gamma}
{(L-s)q/\eta_{\rm s}(q)
+s(q+1)/\eta_{\rm s}(q+1)}.
\label{eq:complete-hardware-rate}
\end{equation}
\end{theorem}

\begin{proof}
We prove optimality by a balancing exchange. Let $k_1,\ldots,k_L$ be the star sizes of an $L$-layer decomposition. They are positive integers with $\sum_i k_i=E_N$. If $k_i\ge k_j+2$, discrete convexity gives
\begin{equation}
f(k_i)+f(k_j)
\ge f(k_i-1)+f(k_j+1).
\label{eq:convex-exchange}
\end{equation}
Thus balancing any pair of sizes that differ by at least two cannot increase the cost. Repeating this exchange at the level of the size profile yields the balanced profile in Eq.~\eqref{eq:division-profile}.

Theorem~\ref{thm:complete} guarantees that $K_N$ has a star decomposition with exactly this profile. Substituting the balanced sizes into Eq.~\eqref{eq:HG} gives Eq.~\eqref{eq:complete-H}, and Eq.~\eqref{eq:hardware-maxmin} gives Eq.~\eqref{eq:complete-hardware-rate}.
\end{proof}

If $\eta_{\rm s}(k)$ is constant, $f$ is linear and every decomposition has the same cost. If $f$ is strictly convex, balanced star sizes are optimal up to their order. A concave range of the measured cost sequence can instead favour an unequal profile.

\begin{table}[htbp]
\caption{Manufacturer-specified splitter costs for the Fiber World SMFOC single-mode FBT family~\cite{FiberWorldCouplers}. The excess losses are typical specifications, not measurements of the deployed apparatus.}
\label{tab:devices}
\begin{ruledtabular}
\begin{tabular}{lccc}
Device&Excess loss (dB)&$\eta_s(k)$&$k/\eta_s(k)$\\
SMFOC 0102 ($k=2$)&0.07&0.984&2.032\\
SMFOC 0104 ($k=4$)&0.30&0.933&4.286\\
SMFOC 0108 ($k=8$)&0.50&0.891&8.976
\end{tabular}
\end{ruledtabular}
\end{table}

Discrete convexity is a device-specific assumption, not a consequence of increasing fan-out. Table~\ref{tab:devices} uses the manufacturer's typical excess losses $\ell_k$ at 1310 or 1550 nm, with $\eta_s(k)=10^{-\ell_k/10}$. The secant slopes between $k=2,4,8$ are $1.127$ and $1.173$, consistent with convexity at those checkpoints. Applying Theorem~\ref{thm:complete-hardware} requires the full splitter-cost sequence, including the intermediate fan-outs, to satisfy Eq.~\eqref{eq:discrete-convex}. These sampled specifications also do not establish the tunable-simplex capability assumed in the theorem.

For one-sided layers, the hardware cost is therefore determined by $k/\eta_s(k)$. For two-sided layers, the two branch probabilities multiply.


\section{Two-sided pair-generation requirements and the eight-user comparison}
\label{sec:twosided-hardware}


For a two-sided layer, the delivered pair flux contains the product of the two branch probabilities.
We first determine the cost of one biclique and then apply it to the eight-user architectures.

Consider a biclique in which the users in $A$ receive one wavelength and the users in $B$ receive its conjugate. Let $p=|A|$ and $q=|B|$. For $u\in A$, let $s_u^+$ be the useful branch probability on the first wavelength; for $v\in B$, let $s_v^-$ be the corresponding probability on the conjugate wavelength. If the two splitter trees have useful transmissions $\eta_+(p)$ and $\eta_-(q)$, then
\begin{equation}
\sum_{u\in A}s_u^+\le\eta_+(p),
\qquad
\sum_{v\in B}s_v^-\le\eta_-(q).
\label{eq:two-sided-splitter}
\end{equation}
For this optimization, assume that the two splitter sides are independently tunable over the full simplices in Eq.~\eqref{eq:two-sided-splitter}, and retain the independent-layer and separable-response assumptions of Sec.~\ref{sec:qkdmodel}. For fixed or constrained splitters, the minimization below must instead be restricted to their actual feasible branch sets. This does not alter the later explicitly fixed-balanced-stage comparison.
The two photons are routed independently and detected without coherent recombination. The wavelength-layer label is retained, so different occurrences of the same edge remain distinguishable.

For target rate $R_{uv}$, write
\[
w_{uv}=\phi_{uv}(R_{uv})
\]
for the required pre-protocol flux. A layer of pair-generation rate $B$ delivers $Bs_u^+s_v^-$ to edge $uv$. We quantify the least pair-generation rate compatible with all edge demands in the biclique by
\begin{equation}
\Theta_{A,B}(\bm R)=
\min_{\substack{s_u^+,s_v^-\ge0\\
\sum_u s_u^+\le\eta_+(p)\\
\sum_v s_v^-\le\eta_-(q)}}
\max_{u\in A,\,v\in B}
\frac{w_{uv}}{s_u^+s_v^-}.
\label{eq:biclique-cost}
\end{equation}
A positive flux requirement divided by a zero branch probability is interpreted as $+\infty$; $0/0$ is taken as zero. If every demand is zero, we set $\Theta_{A,B}=0$ without a logarithmic transformation.

When the bicliques form an edge partition, every edge is supplied by one layer and the layer costs add.

The rate region of an edge-disjoint biclique decomposition $\mathcal P$ is
\begin{equation}
\Rset_{\mathcal P}^{\rm bi}(B_{\rm tot})=
\left\{\bm R\ge0:
\sum_{\Bic(A,B)\in\mathcal P}
\Theta_{A,B}(\bm R)
\le B_{\rm tot}
\right\}.
\label{eq:two-sided-region}
\end{equation}

For fixed branch probabilities, the pair-generation rate of a biclique layer must satisfy $B\ge w_{uv}/(s_u^+s_v^-)$ for every edge in that layer. It is therefore bounded below by the maximum in Eq.~\eqref{eq:biclique-cost}; minimizing over the branch probabilities gives $\Theta_{A,B}$. Summing over edge-disjoint layers proves necessity.


For sufficiency, zero-demand edge constraints may be omitted from the maximum, but the physical side sizes $p,q$ and hence the loss parameters $\eta_+(p),\eta_-(q)$ remain those of the implemented biclique. If at least one demand is positive, the two splitter simplices are compact and the objective diverges whenever a branch required by a positive-demand edge approaches zero, so the minimum is attained. Choosing an optimizing branch allocation in each biclique and assigning the corresponding pair-generation rate $\Theta_{A,B}$ meets every edge demand simultaneously.
\subsection{Repeated edges and labelled streams}

We next allow repeated edges. The eight-layer metropolitan design is a biclique cover rather than a partition: four edges occur twice. A repeated edge can receive flux from both layers, but the two streams may have different visibility, background, or detection response. Their fluxes can be added before the edge protocol only under an explicit homogeneity condition.

 For an occurrence of edge $e=uv$ in layer $\ell$, distinguish the physically delivered flux
\[
x_{\ell,uv}=B_\ell s^+_{\ell u}s^-_{\ell v}
\]
from a bookkeeping demand $w_{\ell e}\ge0$, with $x_{\ell e}\ge w_{\ell e}$. Let $\boldsymbol x_e=(x_{\ell e})_{\ell:e\in E_\ell}$ denote the realized labelled fluxes, and let $F_e^{\rm joint}(\boldsymbol x_e)$ be the rate obtained when the streams are pooled. We say that the repeated occurrences of edge $e$ admit \emph{sum-dependent pooling} when there exists a one-variable response $\widetilde F_e$ such that
\[F_e^{\rm joint}(\mathbf x_e)=\widetilde F_e\!\left(\sum_{\ell:e\in E_\ell}x_{\ell e}\right)\]
for every admissible allocation. Identical conditional states and detection responses are sufficient. If calibration is valid only at fixed proportions, those proportions must be imposed on the realized fluxes, $x_{\ell e}=a_{\ell e}X_e$, where $X_e=\sum_{j:e\in E_j}x_{je}$, $a_{\ell e}\ge0$, and $\sum_\ell a_{\ell e}=1$. Alternatively, an explicit label-resolved discarding rule may define retained fluxes in those proportions, but that rule and its calibration must be stated. Otherwise the full joint response must remain in the optimization.

Let $\mathcal P=\{\Bic(A_\ell,B_\ell)\}_{\ell=1}^{L}$ be a biclique cover. For an edge with sum-dependent pooling, use $F_e=\widetilde F_e$ in Eq.~\eqref{eq:inverse-response} and set $w_e=\phi_e(R_e)$. The allocations $w_{\ell e}\ge0$ must obey
\begin{equation}
\sum_{\ell:e\in E_\ell}w_{\ell e}\ge w_e
\qquad(e\in E).
\label{eq:cover-pooling}
\end{equation}
For layer $\ell$, write $\bm w_\ell=(w_{\ell e})_{e\in E_\ell}$ and let $\Theta_\ell(\bm w_\ell)$ be its minimum source pair-generation rate, obtained from Eq.~\eqref{eq:biclique-cost}.

\begin{proposition}[Repeated-cover flux allocation]
\label{thm:repeated-cover}
Assume that the wavelength-layer label is retained and that every repeated edge admits sum-dependent pooling. If its resolved occurrences are pooled before the edge protocol is applied, the exact simultaneous rate region of a fixed biclique cover is
\begin{equation}
\begin{aligned}
\Rset_{\mathcal P}^{\rm cov}(B_{\rm tot})
=\biggl\{\bm R\ge0:\;&
\min_{\substack{w_{\ell e}\ge0\\
\sum_{\ell:e\in E_\ell}w_{\ell e}\ge\phi_e(R_e)}}
\sum_{\ell=1}^{L}\Theta_\ell(\bm w_\ell)\\
&\le B_{\rm tot}\biggr\}.
\end{aligned}
\label{eq:repeated-cover-region}
\end{equation}
If the labelled streams are processed separately, no homogeneity assumption is needed. Let $F_{\ell e}$ be the calibrated response of the occurrence of edge $e$ in layer $\ell$, and choose a composably secure protocol with parameter $\eps_{\ell e}$ for every active stream. Concatenating the keys gives total rate at least $R_e$ whenever
\begin{equation}
\sum_{\ell:e\in E_\ell}F_{\ell e}(w_{\ell e})\ge R_e,
\label{eq:separate-stream-region}
\end{equation}
and the total distinguishing advantage obeys
\begin{equation}
\eps_{e}^{\rm tot}
\le
\sum_{\ell:e\in E_\ell}\eps_{\ell e}.
\label{eq:separate-stream-security}
\end{equation}
For this composition statement, each stream's security guarantee is understood to include the other streams and all public transcripts as side information. No gain from unresolved coherence between the streams is assumed.
\end{proposition}

In Appendix~\ref{app:repeated-proof}, we separate the allocation of repeated-edge flux across layers from the optimal splitter allocation within each layer. We also prove attainment and composable separate processing there.

If the pooled response depends on how a fixed total flux is divided among occurrences, Eq.~\eqref{eq:repeated-cover-region} is not exact. One must retain $F_e^{\rm joint}(\boldsymbol{x_e})$ or impose the calibrated proportions. Separate processing remains valid with the occurrence-dependent responses $F_{\ell e}$.

The pooled cover problem is generally nonconvex. A single biclique becomes convex after a logarithmic transformation, but forcing several rank-one layer contributions to add to a prescribed edge flux introduces a nonconvex superlevel constraint. For prescribed edge demands, fixed splitter probabilities give a linear program in the layer rates; edge-disjoint bicliques give independent geometric programs.

\subsection{Convex formulation for a single biclique}



For a single biclique, a logarithmic change of variables gives a convex formulation. Keep the physical side sizes $p=|A|$ and $q=|B|$ fixed even when some individual edge demands vanish. 
For nonzero demand, define the active sets $A_+=\{u:\exists v,\ w_{uv}>0\}$ and $B_+=\{v:\exists u,\ w_{uv}>0\}$. Set the branch probabilities of inactive users to zero and introduce logarithmic variables only on the active sets. This removes zero-probability branches that otherwise require infinite logarithmic variables. Write 
\[
s_u^+=\eta_+(p)e^{-x_u},
\qquad
s_v^-=\eta_-(q)e^{-y_v},
\]
and define $t=\log[B\eta_+(p)\eta_-(q)]$. The edge constraints become
\begin{equation}
t\ge \log w_{uv}+x_u+y_v
\qquad (w_{uv}>0),
\label{eq:two-sided-log1}
\end{equation}
while the splitter budgets become
\begin{equation}
\sum_{u\in A_+}e^{-x_u}\le1,
\qquad
\sum_{v\in B_+}e^{-y_v}\le1.
\label{eq:two-sided-log2}
\end{equation}

Minimizing $t$ subject to Eqs.~\eqref{eq:two-sided-log1} and~\eqref{eq:two-sided-log2} is a convex program. Equivalently, the original positive-variable formulation is a geometric program, and both give the same biclique cost. The transmissions remain $\eta_+(p)$ and $\eta_-(q)$ for the implemented splitter size, not those of the reduced active sets. This is a convex formulation in logarithmic variables; the value $\Theta_{A,B}$ need not be convex as a function of the demand matrix. 

For rank-one required flux, $w_{uv}=\alpha_u\beta_v$, with nonnegative factors, the biclique cost has a closed form. Allocating each splitter side in proportion to its factor saturates both splitter budgets and gives
\begin{equation}
\Theta_{A,B}(\bm R)=
\frac{
\left(\sum_{u\in A}\alpha_u\right)
\left(\sum_{v\in B}\beta_v\right)}
{\eta_+(p)\eta_-(q)}.
\label{eq:rankone-biclique-cost}
\end{equation}
Indeed,
\[
\max_u\frac{\alpha_u}{s_u^+}
\ge
\frac{\sum_u\alpha_u}{\sum_u s_u^+}
\ge
\frac{\sum_u\alpha_u}{\eta_+(p)},
\]
and the analogous inequality holds on side $B$. Their product gives the lower bound in Eq.~\eqref{eq:rankone-biclique-cost}, which proportional allocation attains.

We finally specialize to a homogeneous $p\times q$ biclique with linear gain $\gamma$ and common edge rate $R$. In this case, every edge requires flux $R/\gamma$. Hence
\begin{equation}
\Theta_{p,q}(R)=
\frac{Rpq}{\gamma\eta_+(p)\eta_-(q)}.
\label{eq:homogeneous-biclique-cost}
\end{equation}
The numerator is the total required edge flux; the denominator is the product of the useful transmissions of the two splitter trees.

\subsection{Pair-generation requirement of a balanced hierarchy}
\label{subsec:hierarchy-source-cost}

 For the balanced binary hierarchy, the pair-generation requirement can be summed exactly over tree depth. The construction minimizes the number of layers among nonredundant assignments and the receiver load, but the layers near the root split both wavelengths among large user sets. The resulting pair-generation requirement can be evaluated for arbitrary network size.

Let $N=2^h$. Assume that distributing one wavelength to $2^k$ users uses $k$ balanced $1\times2$ stages along each output path and that each stage has useful transmission $\etast\in(0,1]$. Also assume identical linear edge gain $\gamma$ and independently tunable layer pair-generation rates. Then the balanced binary hierarchy has pair-generation rate per unit common key rate
\begin{equation}
\mathcal H_{\rm bal}(N,\etast)
=
\frac{N^2}{4\gamma\etast^{2(h-1)}}
\sum_{d=0}^{h-1}
\left(\frac{\etast^2}{2}\right)^d.
\label{eq:balanced-hierarchy-cost}
\end{equation}
Equivalently,
\begin{equation}
\mathcal H_{\rm bal}(N,\etast)
=
\frac{N^2}{4\gamma\etast^{2(h-1)}}
\frac{1-(\etast^2/2)^h}{1-\etast^2/2}.
\label{eq:balanced-hierarchy-closed}
\end{equation}
At $\etast=1$, this reduces to $\binom N2/\gamma$.

In Appendix~\ref{app:hierarchy-cost-proof}, we derive the result by summing the pair-generation requirement over the depths of the balanced tree.

For fixed $0<\etast<1$, the leading factor scales as
\begin{equation}
\mathcal H_{\rm bal}(N,\etast)
=
\Theta\!\left(
N^{2+2\log_2(1/\etast)}
\right).
\label{eq:balanced-hierarchy-scaling}
\end{equation}
We therefore find that logarithmic receiver load does not imply a quadratic pair-generation requirement once every additional splitter stage has loss.

\begin{figure*}[t]
\centering
\includegraphics[width=0.86\textwidth]{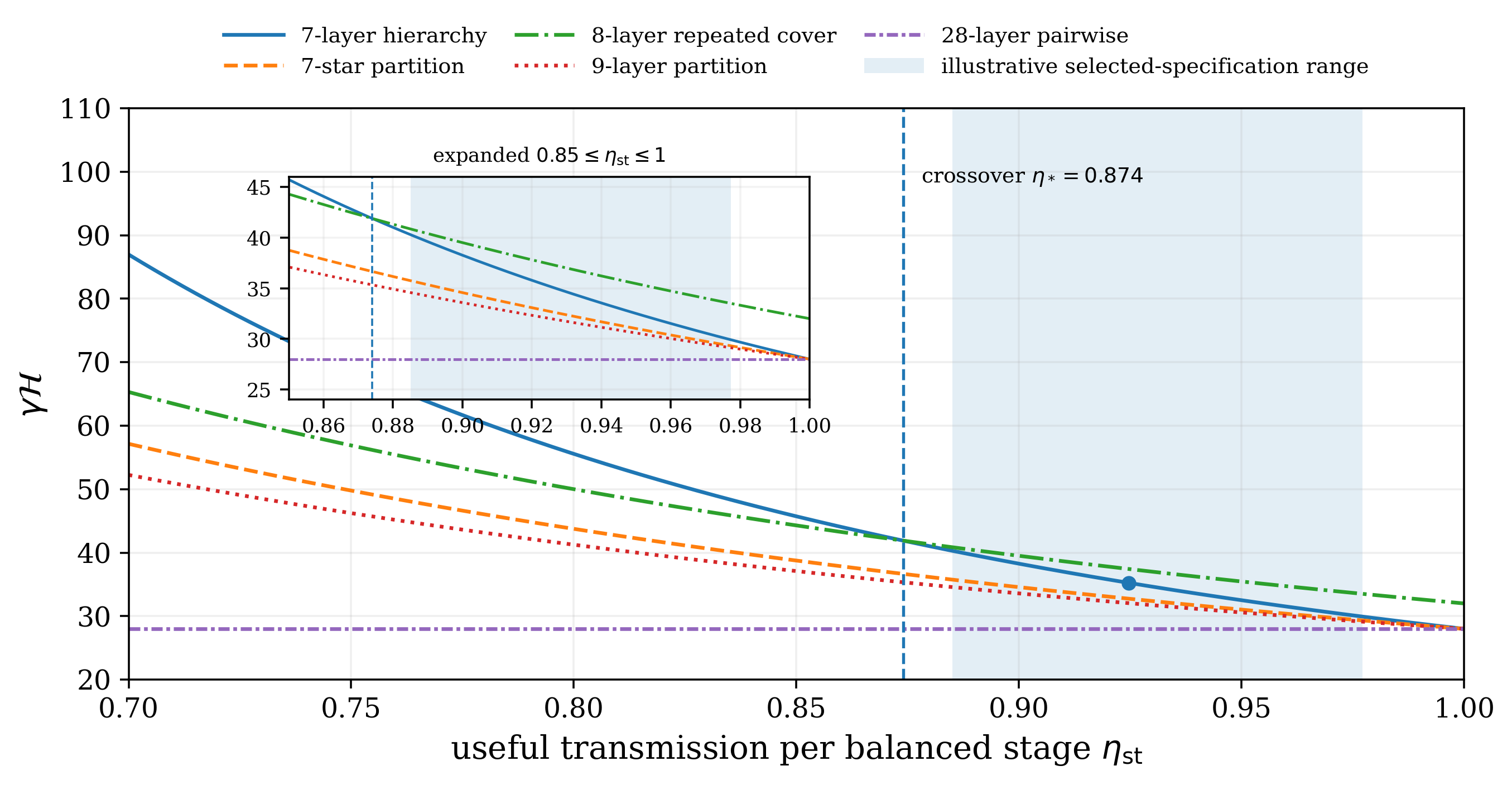}
\caption{Pair-generation requirement for the $K_8$ architectures under the independent-layer model with fixed balanced splitters. The vertical axis shows $\gamma H$, the total input pair-generation requirement per unit common key rate normalized by the common edge gain. The curves are the closed-form requirements in Eqs.~\eqref{eq:k8-Hhier}--\eqref{eq:k8-H9}, together with the unsplit pairwise endpoint $\gamma H=28$. The vertical line marks $\eta_{\rm st}=\sqrt{3-\sqrt5}\simeq0.874$. The shaded interval is $0.885\lesssim\eta_{\rm st}\lesssim0.977$, obtained from the selected stage-loss estimates $0.53$ and $0.10\,\mathrm{dB}$ through $\eta_{\rm st}=10^{-\ell/10}$. The $0.10\,\mathrm{dB}$ endpoint combines the Newport F-CPL-S12155-FCUPC typical excess loss of $0.06\,\mathrm{dB}$ with two $0.02\,\mathrm{dB}$ G.652 fusion-splice losses; the $0.53\,\mathrm{dB}$ endpoint uses the AC Photonics A-grade $3.5\,\mathrm{dB}$ 50/50 maximum output insertion loss, subtracts the ideal $3.0103\,\mathrm{dB}$ division, and adds the same two splice losses~\cite{NewportCoupler,ACPhotonicsCoupler,AFLSplicer}. The filled point is an illustrative evaluation at $0.34\,\mathrm{dB}$, corresponding to $\eta_{\rm st}\simeq0.925$, not a measured apparatus point. The inset expands $0.85\le\eta_{\rm st}\le1$ around the crossover; the shaded interval is illustrative rather than a universal hardware-class range.}
\label{fig:k8-source-crossover}
\end{figure*}

\subsection{Independent-layer source model with fixed balanced splitters}

Figure~\ref{fig:k8-source-crossover} summarizes the source-cost comparison derived in this subsection and marks the crossover in Eq.~\eqref{eq:k8-hierarchy-crossover}.

We now apply the preceding pair-generation formulas to the four fan-out architectures of Proposition~\ref{thm:k8-frontier} and determine how splitter loss changes their ordering by total pair-generation requirement. The pair-generation rate supplied to each occupied layer is adjustable independently. The comparison in this subsection uses fixed balanced $1\times2$ splitter stages: each stage divides the useful output equally between its two branches and has total useful transmission $\etast$. Let a balanced $1\times2$ splitter stage have useful transmission $\etast\in(0,1]$. A four-leaf star uses two stages on one wavelength, while a $2\times2$ biclique uses one stage on each wavelength. In either case, every routed pair crosses two stages and acquires coincidence transmission $\etast^2$.

Let $\mathcal H$ denote the total input pair-generation rate required for one unit of common key rate. The balanced hierarchy contains one $4\times4$ layer, two $2\times2$ layers, and four $1\times1$ layers. Its total pair-generation requirement is
\begin{equation}
\mathcal H_{7}^{\rm hier}
=
\frac{16}{\gamma\etast^4}
+
\frac{8}{\gamma\etast^2}
+
\frac{4}{\gamma}.
\label{eq:k8-Hhier}
\end{equation}
The $4\times4$ root layer crosses two stages on each wavelength, whereas each $2\times2$ layer crosses one stage on each side.

The seven-star architecture contains seven four-edge stars, and therefore
\begin{equation}
\mathcal H_{7}^{1\times4}
=7\frac{4}{\gamma\etast^2}
=\frac{28}{\gamma\etast^2}.
\label{eq:k8-H7}
\end{equation}

Every eight-layer side-size-at-most-two cover consists of eight $K_{2,2}$ layers, and each of those layers contains an edge that appears nowhere else. To see this, suppose one layer contained only repeated edges. Removing it would leave seven $K_{2,2}$ layers covering all $28$ edges. Because those layers contain exactly $28$ edge occurrences, they would form a $K_{2,2}$ partition of $K_8$. This is impossible: every $K_{2,2}$ contributes even degree at each incident vertex, whereas every vertex of $K_8$ has degree seven.

The unrepeated edge in each layer fixes that layer's rate under balanced splitting; the same rate sends equal flux to all four edges. Hence
\begin{equation}
\mathcal H_{8}^{2\times2,\mathrm{rep}}
=8\frac{4}{\gamma\etast^2}
=\frac{32}{\gamma\etast^2}.
\label{eq:k8-H8}
\end{equation}

The nine-layer partition in Eq.~\eqref{eq:k8-nine-bicliques} contains five $K_{2,2}$ layers and four $K_{1,2}$ layers. Its cost is
\begin{equation}
\mathcal H_{9}^{2\times2,\mathrm{nr}}
=5\frac{4}{\gamma\etast^2}
+4\frac{2}{\gamma\etast}
=\frac{20+8\etast}{\gamma\etast^2}.
\label{eq:k8-H9}
\end{equation}

A nonredundant nine-layer assignment with side size at most two has biclique sizes in $\{1,2,4\}$. If $n_j$ is the number of $j$-edge layers, then $n_1+n_2+n_4=9$ and $n_1+2n_2+4n_4=28$. The only nonnegative integer solutions are
\begin{equation}
(n_4,n_2,n_1)=(5,4,0)
\quad\text{or}\quad
(6,1,2).
\label{eq:k8-nine-profiles}
\end{equation}

The latter is impossible: $K_{2,2}$ layers contribute only even degrees, whereas its remaining one $K_{1,2}$ and two $K_{1,1}$ layers can give odd degree to at most six vertices. Since every vertex of $K_8$ has odd degree seven, every nine-layer partition must therefore have five $K_{2,2}$ layers and four $K_{1,2}$ layers. Its total user--layer incidence is
\begin{equation}
5\times4+4\times3=32.
\label{eq:k8-profile-gap}
\end{equation}
Equation~\eqref{eq:biclique-load-bound} requires at least four incidences per user, so all eight users occur exactly four times and $\chi=4$ is forced. Under the homogeneous fixed-balanced-stage model, every such nine-layer partition therefore has the cost in Eq.~\eqref{eq:k8-H9}.

We can now compare the four architectures directly.

\begin{theorem}[Eight-user pair-generation-rate ordering and crossover]
\label{thm:k8-source-phase}Under the homogeneous independent-layer model with fixed balanced $1\times2$ splitter stages,
for $0<\etast\le1$,
\begin{equation}
\mathcal H_{9}^{2\times2,\mathrm{nr}}
\le
\mathcal H_{7}^{1\times4}
\le
\mathcal H_{7}^{\mathrm{hier}},
\label{eq:k8-partition-order}
\end{equation}
with both inequalities strict for $\etast<1$. Moreover,
\begin{equation}
\mathcal H_{7}^{\mathrm{hier}}
\le
\mathcal H_{8}^{2\times2,\mathrm{rep}}
\quad\Longleftrightarrow\quad
\etast\ge\sqrt{3-\sqrt5}\simeq0.874.
\label{eq:k8-hierarchy-crossover}
\end{equation}
\end{theorem}


For output insertion losses $\mathrm{IL}_1$ and $\mathrm{IL}_2$, the useful total splitter transmission before additional common losses is
\[
\eta=10^{-\mathrm{IL}_1/10}+10^{-\mathrm{IL}_2/10}.
\]

Converting a single output loss by subtracting the ideal $3.0103\,\mathrm{dB}$ division presumes balanced outputs; measured branch transmissions are preferable when imbalance matters. For example, the unconnectorized standard P-grade 50/50 coupler in Ref.~\cite{ACPhotonicsCoupler} specifies a maximum insertion loss of $3.3\,\mathrm{dB}$ at each output, giving useful total transmission at least $2\times10^{-3.3/10}\simeq0.935$. This selected specification is above the crossover, but it neither establishes balanced outputs nor characterizes an assembled splitter stage. Other grades and additional connection losses must be evaluated separately.

The statement above charges the splitter stages only. Theorem~\ref{thm:k8-source-phase} treats the spectral-routing hardware as common-mode, or as already absorbed into the calibrated edge gains. If two architectures require different wavelength-routing chains, that assumption can fail. In a concatenated add/drop chain, a channel dropped at position $j$ crosses $j-1$ express ports and one drop port, so its routing loss depends on the drop position. Commercial DWDM add/drop components are supplied against specified insertion-loss grades rather than approximate nominal losses. For the cited AC Photonics 50-GHz device, the manufacturer specifies express-channel insertion loss $\le0.60\,\mathrm{dB}$ and add/drop insertion loss $\le1.0\,\mathrm{dB}$; the cited LIGHTEL PM-DWDM datasheet specifies $0.5\,\mathrm{dB}$ on the express channel and $0.8\,\mathrm{dB}$ (200 GHz) or $1.5\,\mathrm{dB}$ (100 GHz) on the add/drop channel~\cite{ACPhotonicsDWDM,LightelDWDM}. These are manufacturer specification limits rather than typical values: individual components may have lower loss. The seven- and eight-layer designs use 14 and 16 individual wavelength channels, respectively, so this routing cost is not generally common-mode between them. The theorem remains exact for the declared fixed-balanced-splitter model; incorporating architecture-dependent spectral-routing loss is a second optimization rather than a correction to the theorem.

In Appendix~\ref{app:k8-cost-proof}, we prove the ordering and crossover by directly comparing the four closed-form requirements.


The hierarchy minimizes spectrum among nonredundant assignments and receiver load, but its large root splitter can require a higher total pair-generation rate than the bounded-side designs. The unsplit edge-by-edge architecture uses 28 layers and has pair-generation requirement $28/\gamma$, providing the no-fan-out endpoint. Under this splitter-only model, that endpoint has the smallest source requirement; the fan-out designs trade additional source demand for fewer occupied layers.

For a common edge-rate target, each layer rate is set by its weakest routed branch. We finally test whether the seven-versus-nine ordering survives splitter imbalance. Let $b\in(0,1/2]$ be the smaller normalized output fraction of a $1\times2$ splitter, with the same $b$ and $\eta_{\rm st}$ at every stage. The weakest branch then carries fraction $\etast b$ of the input. The two nonredundant costs become
\begin{equation}
\mathcal H_{7}^{1\times4}(b)
=\frac{7}{\gamma\etast^2b^2},
\qquad
\mathcal H_{9}^{2\times2,\mathrm{nr}}(b)
=\frac{5}{\gamma\etast^2b^2}
+\frac{4}{\gamma\etast b}.
\label{eq:k8-imbalance-costs}
\end{equation}
Their ratio is
\begin{equation}
\frac{\mathcal H_{9}^{2\times2,\mathrm{nr}}(b)}
{\mathcal H_{7}^{1\times4}(b)}
=\frac{5+4\etast b}{7}\le1.
\label{eq:k8-imbalance-ratio}
\end{equation}
Equality requires a lossless balanced splitter, $\etast=1$ and $b=1/2$. Any insertion loss or imbalance makes the nine-layer partition strictly cheaper while reducing the maximum receiver load from seven to four. The repeated eight-layer cover does not admit the same reduction because a duplicated edge can receive flux from two layers.

\section{Continuous-wave coincidence model}
\label{sec:cw}


The preceding allocation results assume that each edge response depends only on its delivered flux. In a CW network, shared detector singles generally couple these responses within each fan-out layer.
We first write the true-coincidence contribution for each labeled-layer occurrence. Let $B_\ell\ge0$ be the pair-generation rate of spectral layer $\ell$. For a user $u$ receiving one of its wavelength channels, let $\Sigma_{\ell u}$ be the singles rate of the corresponding detector channel. If edge $e=\{u,v\}$ occurs in layer $\ell$, let $s_{\ell u}$ and $s_{\ell v}$ be the useful routing probabilities of the two conjugate wavelengths. The detected true-coincidence contribution before temporal filtering is
\begin{equation}
C_{\ell e}^{t}
=
B_\ell s_{\ell u}s_{\ell v}
\eta_{\ell u}\eta_{\ell v},
\label{eq:true-cw}
\end{equation}
where $\eta_{\ell u}$ and $\eta_{\ell v}$ include path transmission and detector efficiency at the two endpoint users.

We distinguish true and accidental coincidences explicitly. A true coincidence comes from the two photons of the same generated pair. An accidental coincidence occurs when unrelated detections at the detector channels of the two endpoint users fall within the same coincidence window. If occurrence $(\ell,e)$ is accepted in a window of width $t_{\ell e}^{\rm acc}$, its low-count accidental-coincidence rate is
\begin{equation}
A_{\ell e}
\simeq
\Sigma_{\ell u}\Sigma_{\ell v}
 t_{\ell e}^{\rm acc}.
\label{eq:accidental}
\end{equation}


This approximation is used in the low-count regime $\Sigma_{\ell u}t_{\ell e}^{\rm acc}\ll1$ and $\Sigma_{\ell v}t_{\ell e}^{\rm acc}\ll1$.
Equation~\eqref{eq:accidental} follows the stardard CW entanglement-based QKD model of Ref.~\cite{Neumann2021CWModel}. The rate equations below inherit this approximation and require calibrated source, splitter, and detector inputs.


The accepted window retains only part of the true-coincidence peak. Let $t_{\Delta,\ell e}$ be the FWHM of the relative detection-time distribution for the two endpoint channels. For a Gaussian relative-delay distribution, the capture probability is
\begin{equation}
\kappa_{\ell e}
=
\erf\!\left(
\frac{\sqrt{\ln2}\,t_{\ell e}^{\rm acc}}
{t_{\Delta,\ell e}}
\right).
\label{eq:capture}
\end{equation}
For an edge whose labelled occurrences are pooled under a common calibrated response,
\begin{equation}
C_e^{\rm acc}
=
\sum_{\ell:e\in E_\ell}
\kappa_{\ell e}C_{\ell e}^{t},
\qquad
A_e=
\sum_{\ell:e\in E_\ell}A_{\ell e},
\label{eq:cw-rates}
\end{equation}
and
\begin{equation}
M_e=C_e^{\rm acc}+A_e.
\label{eq:cw-measured}
\end{equation}
Let $e_{{\rm pol},\ell e}$ be the polarization-error probability of an accepted true coincidence from occurrence $(\ell,e)$. Assuming unbiased accidental outcomes, an accidental coincidence has error probability $1/2$,
\begin{equation}
E_e
=
\frac{
\displaystyle\sum_{\ell:e\in E_\ell}
e_{{\rm pol},\ell e}\kappa_{\ell e}C_{\ell e}^{t}
+A_e/2
}{M_e}.
\label{eq:cw-qber}
\end{equation}
If all occurrences of edge $e$ have the same polarization-error probability $e_{{\rm pol},e}$, the numerator reduces to $e_{{\rm pol},e}C_e^{\rm acc}+A_e/2$. 
For the symmetric-noise model used here, we set $E_e^b=E_e^p=E_e$. A device-specific BBM92 security analysis should replace this proxy by the corresponding certified bit- and phase-error bounds.
The asymptotic BBM92 rate is
\begin{equation}
R_e
=
q_eM_e
\left[
1-f_{{\rm EC},e}H_2(E_e^b)-H_2(E_e^p)
\right]_+.
\label{eq:cw-key}
\end{equation}
If $M_e=0$, the rate is set to zero by continuity and the QBER ratio is not formed.



Different wavelength channels are detected separately, but a detector channel can still participate in several pairwise coincidence measurements within the same fan-out layer. Its singles rate then enters the accidental rates of all those links. The CW response is therefore generally joint at the layer level. Reduction to Definition~\ref{def:separable-model} requires the edge-flux sufficiency condition stated in Sec.~\ref{sec:qkdmodel}; otherwise the full layer response must be retained. 

If the same raw detector record is reused in more than one pairwise key, a joint composable security statement additionally requires either a consistent assignment of events to disjoint key streams or a security proof that covers the joint use of those records and all associated transcripts.

\section{Common broadband source model}
\label{sec:published-benchmark}


A single broadband source removes independent control of the occupied layers: one pump drives all occupied conjugate wavelength pairs, and the source spectrum fixes their relative pair-generation rates.
We quantify the source spectrum through calibrated relative pair-generation weights. Let $b_j>0$ be the calibrated relative pair-generation weight of candidate conjugate wavelength pair $j$. The candidate set is finite. A wavelength assignment is an injection $\pi$ from the logical spectral layers of an architecture to distinct candidate pairs. The occupied layer rates are then
\begin{equation}
B_\ell(P,\pi)=P b_{\pi(\ell)},
\label{eq:common-source-layer}
\end{equation}
where $P\ge0$ is a common source-rate scale. The weights $b_j$ may be obtained from a measured channel-resolved source spectrum. No independent rescaling of one occupied layer is allowed within this model. For every occupied detector channel, let $\Sigma_{\ell u}(P,\pi)$ denote its calibrated singles-rate function, including the signal contribution, dark counts, and channel-specific background.


For a fixed architecture and injection $\pi$, we use Eqs.~\eqref{eq:true-cw}--\eqref{eq:cw-key} to determine the joint BBM92 rate vector as a function of $P$. Here the architecture includes its implemented branch ratios; optimizing tunable ratios would require an additional minimization. For prescribed targets $R_e^{\rm tar}$, define

\begin{equation}
P_{\mathcal A}^{\star}(\bm R^{\rm tar})
=
\min_{\pi}\inf\left\{
P\ge0:R_e(P,\pi)\ge R_e^{\rm tar}\ \forall e\in E
\right\}.
\label{eq:common-source-optimization}
\end{equation}


The minimum is taken over the finite set of admissible wavelength injections. Equation~\eqref{eq:common-source-optimization} is the common-source counterpart of the independent-layer pair-generation optimization: it compares the same graph architectures when their layer rates are linked by the source spectrum.

The infimum in Eq.~\eqref{eq:common-source-optimization} is restricted to the operating domain on which the source, detector, and low-count coincidence models have been calibrated; $\inf\emptyset=+\infty$. The quantity $P$ is a common source-rate scale fixed by the calibration convention, not automatically the total occupied pair-generation rate. The latter is
\[
B_{\rm occ}=P\sum_{\ell} b_{\pi(\ell)}.
\]
When the selected spectral weights differ between architectures or injections, minimizing $P$ and minimizing $B_{\rm occ}$ are distinct objectives.

For this deployment, the logical resource graph is $K_4$, realized by the three-layer assignment in Eq.~\eqref{eq:K4stars}, whereas the three-arm star centered at the source, together with the fiber lengths and measured losses, belongs to the physical deployment shown in Fig.~\ref{fig:deployment}. A full numerical evaluation of Eq.~\eqref{eq:common-source-optimization} still requires the channel-resolved source weights, detector-channel backgrounds, timing distributions, and any architecture-dependent spectral-routing losses. We therefore retain Eq.~\eqref{eq:common-source-optimization} as a symbolic benchmark for those remaining apparatus-specific inputs.
\section{Discussion}\label{sec:discussion}

The number of occupied spectral layers is only one part of the cost of passive wavelength sharing. The same wavelength assignment fixes how many channels enter each receiver, how widely each wavelength must be split, and whether requested links appear more than once. Our graph representation makes these resources explicit within the disjoint-recipient, unique-partner routing class of Sec.~\ref{sec:routing_algebra}. For one-sided fan-out, orienting each requested link away from its chosen center and grouping the corresponding edges into fan-out-limited stars gives an exact optimization for an arbitrary network and fan-out limit.
For complete networks, and for complete networks in which each user has one excluded partner, the optimum is governed by two constraints: enough centers are needed to cover the requested links, and the occupied layers must have sufficient total capacity to carry them. Balanced star profiles attain these bounds in the regimes identified above. When the measured splitter-cost sequence $f(k)=k/\eta_s(k)$ is discretely convex, with identical edge gains at each prescribed layer count, same balanced profiles also minimize the total pair-generation requirement for complete networks. Under the tunable-splitter assumptions of Theorem~\ref{thm:complete-hardware}, these profiles therefore minimize both maximum fan-out and source cost at fixed $L$.

Furthermore, allowing both conjugate wavelengths to fan out changes more than just the attainable layer count. It separates spectral demand from receiver wavelength load. For a complete network, the balanced binary hierarchy uses the minimum $N-1$ spectral layers among nonredundant assignments while reducing the maximum receiver load to $\lceil\log_2N\rceil$. Thus, minimum nonredundant spectral demand and logarithmic receiver load can be achieved simultaneously. The eight-user network shows why this does not produce a single preferred architecture. The seven-layer hierarchy has receiver load $\chi=3$, whereas every seven-layer star partition has $\chi=7$. We show that with at most two users on each side of a layer, eight layers require four repeated links. Using nine layers removes these repetitions while keeping the maximum receiver load at four.

In this work, we also show that splitter loss can alter the hierarchy implied by the topology alone. The balanced hierarchy makes this clear. Its receiver load grows only logarithmically with network size, but the layers near the root connect large groups of users. Each wavelength in these layers must therefore pass through several splitter stages before reaching the users, reducing the pair flux delivered to each link. More pair generation is then needed to maintain a fixed link rate. For a fixed splitter transmission ($0<\eta_{\rm st}<1$), the total pair-generation requirement scales as
$\Theta\!\left(N^{2+2\log_2(1/\eta_{\rm st})}\right)$.
Thus, a logarithmic receiver load can still be accompanied by a pair-generation requirement that grows faster than $N^2$.

The eight-user comparison makes this reversal concrete. Under independent control of the occupied spectral layers, the total source rate needed to reach prescribed key-rate targets is obtained from the splitter losses and the dependence of key rate on the pair flux delivered to each link. We show that the architecture with the fewest layers or the smallest receiver load need not minimize this source requirement. In particular, the nine-layer nonredundant bounded-side partition can require less total pair generation than architectures using fewer layers. The seven-layer hierarchy and the eight-layer repeated cover also exchange order at the splitter-transmission threshold derived in Theorem~\ref{thm:k8-source-phase}. The point of this crossover is not that one of these architectures is universally best. It shows that a spectral advantage can disappear, or reappear, as the quality of the passive splitting changes.


The selected unconnectorized coupler specification discussed in Sec.~\ref{sec:twosided-hardware} lies above the crossover, but the ordering for an assembled implementation requires its measured stage transmission and branch balance. Other architectures, such as the nine-layer nonredundant partition, can still have lower splitter-limited source cost. The deployment conclusion must nevertheless keep the model boundary visible: $\eta_s$ charges fan-out, not the wavelength-routing chain.

Spectral routing can supply a second, architecture-dependent cost. In a cascaded add/drop implementation, a wavelength dropped at stage $j$ incurs the pass loss of the preceding $j-1$ devices and the drop loss of the $j$th device. Because the pass and drop ports have different insertion losses, changing the number of occupied layers or their wavelength positions changes this cost. Representative DWDM add/drop specifications put this contribution at several decibels for sufficiently deep cascades~\cite{ACPhotonicsDWDM,LightelDWDM}. In the present analysis we treat this term as common-mode, or as already included in the calibrated edge response. If that is not valid for the hardware being compared, a second routing-cost sequence must be optimized jointly with the splitter profile. This is a natural extension of the same principle established here: the preferred decomposition is fixed by the hardware as well as by the topology.

A single broadband source changes the comparison because the occupied wavelength pairs are generated together rather than controlled independently. Their relative pair-generation rates follow the measured spectrum of the source, while changing the overall source brightness scales these rates together. The coincidence measurements are also coupled within a fan-out layer. Although different wavelength channels are detected separately, the same detector channel for one user can contribute to the coincidence measurements of several user pairs. Its singles rate, therefore, affects the accidental-coincidence rate of more than one link. True coincidences come from the two photons of the same generated pair, whereas accidental coincidences arise when unrelated detections fall within the same coincidence window. Equation~\eqref{eq:common-source-optimization} includes both the fixed spectral distribution of the source and this coupling between the coincidence measurements. Therefore, the resulting comparison between architectures depends on the actual setup and requires measured detector noise, coincidence timing, and splitter performance.

The results are exact at different stages of the analysis. The graph-theoretic results give exact limits on the wavelength architectures and their resources. The corresponding key-rate regions are exact under the stated physical assumptions when the relation between delivered pair flux and key rate is known exactly. If the key-rate model provides only a guaranteed lower bound, our calculation gives rates that are certainly achievable, but the actual network may support higher rates. In the CW model, one wavelength channel received by a user can be correlated with the conjugate wavelength sent to several other users. The same time-tag record is therefore used in several pairwise coincidence calculations. Because the accidental-coincidence rate for each link depends on the singles rates of its two detection records, these link calculations are coupled, and we include this dependence explicitly. Equation~\eqref{eq:accidental} uses the low-count approximation for the accidental coincidences. The graph results are fixed by the wavelength architecture, whereas the predicted key rates also depend on the parameters of the experimental setup.

\section{Conclusion}


The results developed here turn the wavelength assignment of a passive entanglement network into an exact optimization on its connectivity graph. Given the set of links a network must provide, the decomposition of that graph into spectral layers fixes the spectral demand, the receiver wavelength load, the fan-out required of each splitter, and whether any link is repeated. These quantities are not independent, and the exact laws obtained above determine which combinations of them are attainable.

Taken together, the results form a design procedure for the network element that distributes entanglement---the quantum network provider (QNP) used in the deployment reported in the companion paper~\cite{companionB}. Given a requested topology and the fan-out available from the hardware, Theorem~\ref{thm:master-one-sided} gives the exact one-sided optimization for an arbitrary network, while Theorem~\ref{thm:complete} and Corollary~\ref{thm:cocktail} give explicit laws for complete networks and for complete networks with one excluded partner. Allowing both conjugate wavelengths to fan out, Theorem~\ref{thm:tree} supplies an assignment that attains the minimum layer count among nonredundant assignments while holding the maximum receiver load to $\lceil\log_2N\rceil$; Proposition~\ref{thm:block-hierarchy} gives the bounded-side construction, and Proposition~\ref{thm:k8-frontier} gives the exact eight-user resource points. For complete meshes with identical edge gains, when the measured splitter-cost sequence $f(k)=k/\eta_{\rm s}(k)$ satisfies the discrete-convexity condition, Theorem~\ref{thm:complete-hardware} identifies the layer-size profile that minimizes the total pair-generation requirement. For the eight-user comparison under fixed balanced $1\times2$ splitter stages, Theorem~\ref{thm:k8-source-phase} determines which of the compared architectures has lower source cost at a given splitter transmission. Finally, the continuous-wave coincidence model converts the resulting flux into BBM92 secret-key rates under its stated noise and security assumptions, using the source spectrum and detector conditions of the apparatus in use.

A design chosen this way is therefore fixed by the hardware, and not by the topology alone. The architecture that saves the most spectrum need not minimize the source rate required to operate the network, and the ordering of architectures can change with splitter transmission, source spectrum, and receiver conditions. This is why the procedure ends with measured parameters rather than with a graph: a spectral saving is useful only when the associated splitter loss and pair-generation requirement still allow the target key rates to be reached.

Several extensions follow naturally. The one-sided optimization already applies to an arbitrary requested topology, whereas the two-sided  constructions and bounds developed here concern complete networks; extending them to general graphs would cover networks in which only some user pairs require a link. The key-rate analysis is asymptotic, and a finite-key treatment would place the comparison between architectures on an operationally composable footing. The present source-cost model also isolates fan-out from spectral-routing loss; incorporating the pass/drop cost sequence of a cascaded DWDM or add/drop chain would turn the hardware comparison into a joint splitter-and-routing optimization. Finally, passive routing and active reconfiguration are complementary rather than competing: an architecture that is optimal under one set of link conditions need not remain so when those conditions change, and combining the exact resource laws obtained here with switched entanglement distribution~\cite{Wang2026Dynamic} is a natural route to networks that retain minimal passive resource use while adapting to the conditions under which they operate.

\begin{acknowledgments}
This work was carried out within the project skQCI--Slovak Quantum Communication Infrastructure, funded by the European Union under the Digital Europe Programme, call \mbox{DIGITAL-2021-QCI-01-DEPLOY-NATIONAL}, grant agreement No.~101091548, with 50\,\% national co-financing from the Recovery and Resilience Plan of Slovakia. The project forms part of the European Quantum Communication Infrastructure (EuroQCI) initiative. We thank David Polzoni for valuable discussions. M.Z. acknowledges support from the R4 scholarship 09I03-03-V04-00777 (QENTAPP). D.A. acknowledges support from fellowship No.~1156/01/01, funded by the European Union's Horizon~2020 research and innovation programme under the Marie Sk\l{}odowska-Curie grant agreement No.~945478. Views and opinions expressed are those of the authors only and do not necessarily reflect those of the European Union or the European Commission. Neither the European Union nor the granting authority can be held responsible for them.
\end{acknowledgments}

\section*{Author Contributions}

E.P. contributed to conceptualization, methodology, formal analysis,
investigation, software, writing -- original draft, and writing --
review and editing. G.B. contributed to conceptualization,
investigation, and writing -- review and editing. S.S. contributed to
conceptualization, investigation, and writing -- review and editing.
K.K. contributed to investigation and writing -- review and editing.
S.P. contributed to investigation and writing -- review and editing.
P.R. contributed to supervision and writing -- review and editing.
M.Z. contributed to supervision, writing -- review and editing, and
funding acquisition. D.A. contributed to conceptualization,
supervision, writing -- review and editing, and funding acquisition.

\appendix
\section{Supporting graph-theoretic proofs}
\label{app:graph-proofs}

The main text retains the statements and physical interpretation of the central resource laws. Here we collect the longer graph-theoretic verifications.

\subsection{Proof of Lemma~\ref{thm:dense}: Minimum-degree condition for balanced star profiles}
\label{app:dense-proof}
\begin{proof}
Tarsi's prescribed-star theorem states that a simple graph on $n$ vertices with minimum degree at least $n/2+r-1$ admits a decomposition into any prescribed positive star sizes summing to its edge count and bounded by $r$~\cite{Tarsi1981Stars,CameronHorsley2020Stars}. Apply it to the list consisting of $L-s$ copies of $q$ and $s$ copies of $q+1$. The list has total $m$ and largest entry $r_L$, so Eq.~\eqref{eq:dense-cond} ensures an edge-disjoint star decomposition with these sizes. Its largest star meets Eq.~\eqref{eq:basic-one-sided-bounds}, proving Eq.~\eqref{eq:dense-law}.
\end{proof}

\subsection{Proof of Corollary~\ref{thm:cocktail}: One excluded partner per user}
\label{app:cocktail-proof}

\begin{proof}
 The lower bounds follow from Eq.~\eqref{eq:cp-range}. Choose the endpoints of one deleted matching edge as the two users that never act as centers. Every other user remains adjacent to both of them. The remaining $N-2$ users induce an $(N-4)$-regular graph; since this degree is even, orient each component along an Euler tour, giving out-degree $(N-4)/2$ at every remaining user. The star centered at such a user contains its outgoing neighbors together with the two noncenters, and therefore has size $2+(N-4)/2=N/2$. These $N-2$ stars partition all requested edges.

Splitting one star gives $L=N-1$. Since $\left\lceil N(N-2)/[2(N-1)]\right\rceil=N/2$ for even $N\ge4$, this split still attains the lower bound on the largest star. For $L\ge N$, the balanced largest size is at most $N/2-1$, while $\delta(CP_N)=N-2$. Lemma~\ref{thm:dense} therefore applies and gives Eq.~\eqref{eq:cp-rho}. Inversion gives Eq.~\eqref{eq:cp-L}.
\end{proof}

A balanced profile also exists at the remaining case $L=N-1$, although merely splitting one of the preceding stars need not produce it. Write $N=2k$ and choose one noncenter $z$, whose excluded partner is $z'$. Start from a regular tournament on the other $2k-1$ vertices. Delete the $k-1$ matching edges not incident on $z$, and orient every existing edge incident on $z$ toward $z$. Each deleted matching edge lowers the out-degree of exactly one distinct vertex by one. The resulting nonempty stars have sizes $k$ at $k-1$ centers and $k-1$ at the other $k$ centers, including $z'$. Their sizes sum to $2k(k-1)$ and differ by at most one.

\subsection{Proof of Proposition~\ref{thm:block-hierarchy}: Bounded-side block hierarchy}
\label{app:block-proof}

\begin{proof}
We partition the users into $g$ nonempty blocks $V_1,\ldots,V_g$ of sizes $n_i\le r$. For every $i<j$, include $\Bic(V_i,V_j)$. These $\binom g2$ layers partition the edges joining different blocks. Inside each $V_i$, use the balanced binary-tree partition of $K_{n_i}$, which has $n_i-1$ layers and maximum local incidence $\lceil\log_2n_i\rceil$. The internal partitions and cross-block bicliques are edge disjoint, so the full assignment is nonredundant. Its layer count is
\begin{equation}
\binom g2+\sum_{i=1}^{g}(n_i-1)
=
\binom g2+N-g.
\end{equation}
A user in block $V_i$ occurs in the $g-1$ cross-block layers incident on $V_i$ and in at most $\lceil\log_2n_i\rceil\le\lceil\log_2r\rceil$ internal layers. Finally, $g=N/r+O(1)$ gives Eq.~\eqref{eq:block-hierarchy-asymptotic}.
\end{proof}

\subsection{Proof of Proposition~\ref{thm:k8-frontier}: Eight-user resource points}
\label{app:k8-proof}

\begin{proof}
We prove the five claims in order. The first statement follows from Eqs.~\eqref{eq:complete-bp} and~\eqref{eq:complete-local-biclique} and Theorem~\ref{thm:tree}. Relabel the users by $0,\ldots,7$ and write $XY|ZW$ for $\Bic(\{X,Y\},\{Z,W\})$. An explicit balanced hierarchy is
\begin{equation}
0123|4567,
\quad 01|23,
\quad 0|1,
\quad 2|3,
\quad 45|67,
\quad 4|5,
\quad 6|7.
\label{eq:k8-hierarchical-bicliques}
\end{equation}
The root layer covers all edges between the two four-user halves. The remaining layers partition the edges inside each half. Every user appears in the root layer, one $2\times2$ layer, and one $1\times1$ layer, hence $\chi=3$.

For the seven-star construction, label seven users by $\mathbb Z_7$ and the eighth by $\infty$. Direct each edge on $\mathbb Z_7$ from $i$ to $i+1,i+2,i+3$ modulo seven and take
\begin{equation}
S_i=S\bigl(i,\{\infty,i+1,i+2,i+3\}\bigr).
\label{eq:k8-seven-stars}
\end{equation}
These seven stars of fan-out four partition all $28$ edges. In any seven-star partition, at most seven users act as centers. Two users cannot both be noncenters, because the edge between them would belong to no star. Hence exactly one user is a noncenter, and each of the remaining seven users centers exactly one of the seven stars. The noncenter $z$ must be a leaf in the star centered at each of its seven neighbors. Therefore $z$ appears in every layer and $\chi=7$.

For side size at most two, Eq.~\eqref{eq:biclique-load-bound} gives $\chi_u\ge4$ for every user. The eight users therefore require at least $32$ user--layer incidences, while one layer contains at most four users. Hence at least eight layers are necessary. At equality, every user occurs four times and every layer is a $K_{2,2}$ containing four edge occurrences. The total is $8\times4=32$, so
\begin{equation}
\Delta=32-28=4,
\qquad
\chi=4.
\label{eq:k8-eight-forced}
\end{equation}
Each user then has eight incident edge occurrences but only seven distinct partners, so its excess degree is one. The repeated edges therefore form a perfect matching: exactly four distinct links occur twice.
The cover
\begin{align}
&AE|BF,
\quad AE|CG,
\quad AE|DH,
\quad BF|CG,
\nonumber\\
&BF|DH,
\quad CG|DH,
\quad AF|BE,
\quad CH|DG
\label{eq:k8-eight-cover}
\end{align}
attains the bound. The edges $AB$, $CD$, $EF$, and $GH$ occur twice; all others occur once.

Equation~\eqref{eq:k8-eight-forced} rules out a nonredundant eight-layer assignment with side size at most two. The nine bicliques
\begin{align}
&1|02,
\quad 2|36,
\quad 6|17,
\quad 7|45,
\nonumber\\
&01|47,
\quad 04|23,
\quad 04|56,
\quad 16|35,
\quad 23|57
\label{eq:k8-nine-bicliques}
\end{align}
partition all $28$ edges, and every user appears in four layers.

Finally, a $1\times1$ layer carries one edge, so all $28$ edges require separate layers. Every user is incident on seven of them.
\end{proof}

\section{Supporting pair-generation proofs}
\label{app:physical-proofs}

Here we collect the longer optimization and algebraic proofs used in the source-allocation analysis.

\subsection{Proof of Proposition~\ref{thm:hardware-region}: Tunable fan-out-dependent splitters}
\label{app:hardware-proof}

\begin{proof}
Let $B_S$ be the input pair-generation rate of star $S$, and let $x_e=B_Ss_e$ be the flux delivered to edge $e$. Equation~\eqref{eq:fan-out-budget} gives
\begin{equation}
\sum_{e\in S}x_e
\le B_S\eta_{\rm s}(k_S).
\label{eq:hardware-layer-flux}
\end{equation}
Since $x_e\ge\phi_e(R_e)$,
\begin{equation}
B_S\ge
B_S^{\min}(\bm R)
:=
\frac{1}{\eta_{\rm s}(k_S)}
\sum_{e\in S}\phi_e(R_e).
\label{eq:min-star-pair-generation-rate}
\end{equation}
Summing these bounds proves necessity.

For sufficiency, assign star $S$ the pair-generation rate in Eq.~\eqref{eq:min-star-pair-generation-rate} and choose
\begin{equation}
s_e=
\eta_{\rm s}(k_S)
\frac{\phi_e(R_e)}{
\sum_{f\in S}\phi_f(R_f)}.
\label{eq:hardware-branch-choice}
\end{equation}
Then $B_Ss_e=\phi_e(R_e)$ and the useful splitter budget is saturated. Stars with zero demand are assigned zero pair-generation rate. Hence every point in the stated region is attainable.
\end{proof}

\subsection{Proof of Proposition~\ref{thm:repeated-cover}: Repeated-cover flux allocation}
\label{app:repeated-proof}

\begin{proof}
We separate the cross-layer bookkeeping allocation from the within-layer physical splitter optimization. In layer $\ell$, let $w_{\ell e}$ denote the minimum flux allocated to edge $e$, while for $e=uv$ the physically delivered flux is
\begin{equation}
x_{\ell,uv}=B_\ell s^+_{\ell u}s^-_{\ell v},
\qquad x_{\ell e}\ge w_{\ell e}\ge0.
\label{eq:induced-cover-flux}
\end{equation}
Sum-dependent pooling converts the target rate $R_e$ into the requirement $\sum_{\ell:e\in E_\ell}w_{\ell e}\ge\phi_e(R_e)$. For a chosen bookkeeping allocation $\bm w_\ell$, layer $\ell$ requires pair-generation rate at least $\Theta_\ell(\bm w_\ell)$, where the within-layer optimization enforces $x_{\ell e}\ge w_{\ell e}$. Summing the layer costs proves necessity of Eq.~\eqref{eq:repeated-cover-region}.

Because each $\Theta_\ell$ is nondecreasing in every allocated demand, the bookkeeping variables may be restricted to
\begin{equation}
\sum_{\ell:e\in E_\ell}w_{\ell e}=\phi_e(R_e)
\label{eq:cover-equality}
\end{equation}
for every edge with finite demand. These bookkeeping variables lie in a compact product of intervals. The layer cost is lower semicontinuous: compactness of the splitter simplices gives convergent subsequences, and the physical constraints $B s_u^+s_v^-\ge w_{uv}$ remain valid in the limit. Hence the sum of layer costs attains its minimum. An optimizing physical splitter may oversupply some entries, but monotonicity of the pooled response ensures that this oversupply preserves every target. This proves sufficiency.

For separate processing, the layer label remains in the classical transcript. Apply the calibrated protocol $F_{\ell e}$ to each realized stream and concatenate the resulting keys. Rates add as in Eq.~\eqref{eq:separate-stream-region}, while sequential composition and the triangle inequality give Eq.~\eqref{eq:separate-stream-security}~\cite{BenOr2005Composable}, provided each stream's security definition includes the other streams and all public transcripts as side information.
\end{proof}
\subsection{Balanced-hierarchy requirement}
\label{app:hierarchy-cost-proof}

We sum the requirement level by level. At depth $d$, the tree has $2^d$ internal nodes. Each child subtree contains $N/2^{d+1}$ users, so one layer carries a biclique with
\begin{equation}
p=q=\frac{N}{2^{d+1}}
\end{equation}
and $pq=N^2/4^{d+1}$ edges. Each wavelength passes through $h-d-1$ splitter stages, giving coincidence transmission $\etast^{2(h-d-1)}$. Equation~\eqref{eq:homogeneous-biclique-cost} therefore assigns the depth-$d$ layers total cost
\begin{equation}
\frac{N^2}{4\gamma\etast^{2(h-1)}}
\left(\frac{\etast^2}{2}\right)^d.
\end{equation}
Summing over $d=0,\ldots,h-1$ gives Eq.~\eqref{eq:balanced-hierarchy-cost}. Setting $\etast=1$ gives the geometric sum $2(1-1/N)$ and hence $N(N-1)/(2\gamma)$.

\subsection{Proof of Theorem~\ref{thm:k8-source-phase}: Eight-user ordering and crossover}
\label{app:k8-cost-proof}

\begin{proof}
We compare the closed-form requirements pairwise. Subtracting the seven-star pair-generation requirement from Eq.~\eqref{eq:k8-Hhier} gives
\begin{equation}
\mathcal H_{7}^{\mathrm{hier}}
-
\mathcal H_{7}^{1\times4}
=
\frac{4(1-\etast^2)(4-\etast^2)}{\gamma\etast^4}\ge0.
\end{equation}
The other adjacent difference needed for Eq.~\eqref{eq:k8-partition-order} is
\begin{equation}
\mathcal H_{7}^{1\times4}
-
\mathcal H_{9}^{2\times2,\mathrm{nr}}
=
\frac{8(1-\etast)}{\gamma\etast^2}\ge0.
\end{equation}

Finally, comparing Eqs.~\eqref{eq:k8-Hhier} and~\eqref{eq:k8-H8} gives
\begin{equation}
\etast^4-6\etast^2+4\le0.
\end{equation}
On $0<\etast\le1$, this is equivalent to $\etast^2\ge3-\sqrt5$.
\end{proof}

\bibliography{Bibliography}

\end{document}